\documentclass[onecolumn,floatfix,showpacs,aps,nofootinbib]{revtex4}

\usepackage[dvips]{epsfig}
\usepackage[english]{babel}

\usepackage{bbm}
\usepackage{amsthm} %pro ambiente PROOF

\usepackage{array}
\usepackage{amsmath}
\usepackage{multirow}
\usepackage{amsfonts}
\usepackage{amssymb}
\usepackage{xcolor}
\usepackage{graphicx,color}
\usepackage{hyperref}
\hypersetup{colorlinks=true,linkbordercolor=blue,linkcolor=blue, citecolor=blue}
\usepackage{subeqnarray}
\usepackage{cancel}
\usepackage{wasysym}
\usepackage{indentfirst} %primeiro paragrafo com espaço

\makeatletter
\renewcommand*\l@section{\@dottedtocline{1}{1.5em}{2.3em}}
\makeatother
\newtheorem{proposition}{Proposition}
\newtheorem{theorem}{Theorem}
\newtheorem{lemma}{Lemma}
\newtheorem{definition}{Definition}

\begin{document}

\title{The Restricted Inomata-McKinley spinor-plane, homotopic deformations and the Lounesto classification}

\author{D. Beghetto$^{1}$} \email{dbeghetto@feg.unesp.br}
\author{R. J. Bueno Rogerio$^{1,2}$} \email{rodolforogerio@feg.unesp.br}
\author{C. H. Coronado Villalobos$^{3}$} \email{ccoronado@id.uff.br}
\affiliation{ \\$^1$Universidade Estadual Paulista (UNESP)\\Faculdade de Engenharia de Guaratinguet\'a, Departamento de F\'isica e Qu\'imica\\
12516-410, Guaratinguet\'a, SP, Brazil}
\affiliation{$^{2}$Instituto de F\'isica e Qu\'imica, Universidade Federal de Itajub\'a - IFQ/UNIFEI, \\
Av. BPS 1303, CEP 37500-903, Itajub\'a - MG, Brasil.}
\affiliation{ \\$^3$Instituto Nacional de Pesquisas Espaciais (INPE)\\Ci\^encias Espaciais e Atmosf\'ericas\\
12227-010, S\~ao Jos\'e dos Campos, SP, Brazil}
%\\ $^2$Departamento de Ci\^encias Exatas e da Terra,
%Universidade Federal de S\~ao Paulo (UNIFESP),\\
%09972-270, Diadema, SP, Brazil.}

%\date{\today}

\begin{abstract}
 We define a two-dimensional space called the spinor-plane, where all spinors that can be decomposed in terms of Restricted Inomata-McKinley (RIM) spinors reside, and describe some of its properties. Some interesting results concerning the construction of RIM-decomposable spinors emerge when we look at them by means of their spinor-plane representations. We show that, in particular, this space accomodates a bijective linear map between mass-dimension-one and Dirac spinor fields. As a highlight result, the spinor-plane enables us to construct homotopic equivalence relations, revealing a new point of view that can help to give one more step towards the understanding of the spinor theory. In the end, we develop a simple method that provides the categorization of RIM-decomposable spinors in the Lounesto classification, working by means of spinor-plane coordinates, which avoids the often hard work of analising the bilinear covariant structures one by one.
\end{abstract}

\maketitle

\section{Introduction}

 The so called Inomata-McKinley spinors are a particular class of solutions of the non-linear Heisenberg equation \cite{akira}. A subclass of Inomata-McKinley spinors called restricted Inomata-McKinley (RIM) spinors was revealed to be useful in describing neutrino physics \cite{novello}. It is well known that free linear massive (or mass-less) Dirac fields can be represented as a combination of RIM-spinors \cite{novello}. Moreover, it was recently shown \cite{RIM} that such Dirac spinors are necessarily type-1 in the so-called Lounesto classification, and that they are all non exotic spinors, i.e., the spacetime itself needs to have an underlying trivial topology\footnote{In this work, what we call by a manifold with underlying trivial topology is a manifold $M$ that has a trivial fundamental group $\pi_1(M) = 0$. Otherwise, the manifold $M$ will be said to have a non-trivial topology.} in order to enable the very existence of RIM-spinors. Thus, the decomposition in terms of RIM-spinors itself is not allowed in a spacetime with non-trivial topology.

 The Elko eigenspinors of charge conjugation operator, which are mass-dimension-one (MDO) spinors\footnote{We will use ``MDO'' to call these spinor fields.}, compose a new set of spinors with an interesting and complex structure on its own \cite{jcap, 1305}. MDO spinors form a complete set of eigenspinors of the charge conjugation operator, $C$, however, they have dual helicity and can take positive (self-conjugated) and negative (anti-self-conjugated) eigenvalues of $C$, contrasting with the Majorana, which take only the positive value and carry single-helicity. From the physical point of view, such spinors are constructed to be ``invisible" to other particles, once all the couplings with the fields of the Standard Model are not allowed, except for the Higgs boson, thus, becoming a natural candidate to describe dark matter \cite{jcap}.

 The idea of mapping MDO and Dirac spinor fields is not new \cite{map1, map2, Julio}. However, the works developed towards this proposal use MDO as being a type-5 spinor field within Lounesto classification, taking the bilinear covariants associated to this class as fundamental elements in the construction of the mapping. It is well known that MDO fields do not fulfill the requisites to fit in the Lounesto Classification (for more details the reader is referred to \cite{bilineares}), since their dual is defined in a different way than the usually is imposed in such classification. Moreover, MDO fields are governed by a whole non-usual dynamics, carrying a new and different physical content. Then, a true map between Dirac and MDO spinors is, in fact, a map between different spinor spaces. Thus, to transcend the need of using the bilinear structures associated to the spinors would be welcome in such attempt to construct the aforementioned map.

 Topology is a very important field of study not only in Mathematics, but also in many areas of Physics. By means of its methods and concepts, Topology often allows the discovery and a deep understanding of several substantial aspects in condensed matter, cosmology and many other fields. Furthermore, in particular by homotopical tools, interesting connections between only apparently disconnected areas and results are often revealed, which makes this field of study so powerful and interesting.

 In the present work, we construct a space called spinor-plane, which is a two-dimensional space with its elements being every spinor that can be written in terms of RIM-spinors. The study of this space leads to a better understanding of properties and relations between these spinors, as we shall see. The fundamental concept, in Algebraic Topology, of homotopic maps reveals impressive in the study of the spinors in this plane. Moreover, by means of the spinor-plane, we provide a truthful and direct categorization of RIM-decomposable spinors in the so-called Lounesto classification of spinor fields. Also, we show that an easily constructible bijective map between Dirac and MDO spinors is a direct result of the properties of the spinor-plane, dealing only with their decompositions in terms of RIM-spinors.

 This paper is organized as follows: A short elementary review on the Lounesto classification and on RIM-spinors is presented in two separated Subsections in Section \ref{prelude}. The decomposition of MDO spinors in terms of RIM-spinors is made in Section \ref{MDORIM}. In Section \ref{spinorplane} we construct the two-dimensional space of all RIM-decomposable spinors and present some of its properties, with the main results being shown as two \textit{Lemmas}. Strong results relating homotopy and RIM-decomposable spinors are condensed in two Theorems rigorously constructed in Section \ref{homotopic}. In Section \ref{bilinears} we devote our attention to the bilinear covariants of the particular class of RIM-decomposable spinors that has its adjoint defined in the Dirac fashion, with the results presented as two Propositions. In the last Section we conclude.

\section{Elementary review}\label{prelude}
This section is reserved for a small review on the introductory elements that are necessary for the study carried out in the scope of this paper.

\subsection{The Lounesto Classification}\label{subLounesto}

Let $\psi$ be an arbitrary spinor field, belonging to a section of the vector bundle $\mathbf{P}_{Spin^{e}_{1,3}}(\mathcal{M})\times\, _{\rho}\mathbb{C}^4$, where $\rho$ stands for the entire representation space $D^{(1/2,0)}\oplus D^{(0,1/2)}$. The usual bilinear covariants associated to $\psi$ reads 
\begin{eqnarray}
\label{Azao} A & = &\bar{\psi}\psi, \;\mbox{(scalar)}\\  
\label{Bzao}B & = & i\bar{\psi}\gamma_5\psi, \;\mbox{(pseudo-scalar)} \\ 
\mathbf{J}&=&J_\mu \theta^\mu = \bar{\psi} \gamma_\mu \psi \theta^\mu, \;\mbox{(vector)} \\
\mathbf{K}&=& K_\mu \theta^\mu = \bar{\psi} i \gamma_{0123} \gamma_\mu \psi \theta^\mu, \;\mbox{(axial-vector)}\\ 
\label{Szao} \mathbf{S} &=& S_{\mu \nu} \theta^{\mu \nu} = \frac{1}{2} \bar{\psi} i \gamma_{\mu \nu} \psi \theta^\mu \wedge \theta^\nu, \;\mbox{(bi-vector)} %\mathbf{S}&=&\textcolor{red}{S_{\mu\nu}\gamma^{\mu}\wedge\gamma^{\nu}} = \frac{1}{2}\bar{\psi}\textit{i}\gamma_{\mu\nu}\psi\gamma^{\mu}\wedge\gamma^{\nu}, \;\mbox{(bi-vector)}
\end{eqnarray} 
where $\gamma_{0123}:=\gamma_5=\gamma_0\gamma_1\gamma_2\gamma_3$ and $\gamma_{\mu\nu} : = \gamma_{\mu}\gamma_{\nu}$. Denoting by $\eta_{\mu \nu}$ the Minkowski metric, the set $\{\mathbbm{1},\gamma
_{I}\}$ (where $I\in\{\mu, \mu\nu, \mu\nu\rho, {5}\}$ is a composed index) is a basis for the Minkowski spacetime
${\cal{M}}(4,\mathbb{C})$ satisfying  $\gamma_{\mu }\gamma _{\nu
}+\gamma _{\nu }\gamma_{\mu }=2\eta_{\mu \nu }\mathbbm{1}$, and $\bar{\psi}=\psi^{\dagger}\gamma_{0}$ stands for the adjoint spinor with respect to the Dirac dual. Yet, the elements $\{ \theta^\mu \}$ are the dual basis of a given inertial frame $\{ \textbf{e}_\mu \} = \left\{ \frac{\partial}{\partial x^\mu} \right\}$, with $\{x^\mu\}$ being the global spacetime coordinates. Also, we are denoting $\theta^{\mu \nu} := \theta^\mu \wedge \theta^\nu$.

In the Dirac theory, the above bilinear covariants are interpreted respectively  as the  mass of the particle ($\sigma$), the pseudo-scalar ($\omega$) relevant for parity-coupling, the current of probability ($\mathbf{J}$), the direction of the electron spin ($\mathbf{K}$), and the probability density of the intrinsic electromagnetic moment ($\mathbf{S}$) associated to the electron. In general grounds, it is always expected
to associate such bilinear structures to physical observables.

The bilinear forms defined in (\ref{Azao})-(\ref{Szao}) obey the so-called Fierz-Pauli-Kofink (FPK) identities, given by \cite{baylis}
\begin{eqnarray}\label{fpkidentidades}
\label{6}\boldsymbol{J}^2 & = & A^2+B^2, \\
J_{\mu}K_{\nu}-K_{\mu}J_{\nu} & = & -B S_{\mu\nu} - \frac{A}{2}\epsilon_{\mu\nu\alpha\beta}S^{\alpha\beta}, 
\\
J_{\mu}K^{\mu} & = & 0, \\ 
\label{9}\boldsymbol{J}^2 & = & -\boldsymbol{K}^2.
\end{eqnarray}
So, the algebraic constraints presented in \eqref{Azao}-\eqref{Szao} reduce the possibilities of (only) six different spinor classes  (for which $\boldsymbol{J}$ is always non-null), known as \emph{Lounesto Classification} \cite{Lou}:
\begin{enumerate}
  \item $A\neq0$, $\quad B\neq0$;
  \item $A\neq0$, $\quad B=0$;
  \item $A=0$, $\quad B\neq0$;
  \item $A=0=B,$ \hspace{0.5cm} $\textbf{K}\neq0,$ $\quad\textbf{S}\neq0$;
  \item $A=0=B,$ \hspace{0.5cm} $\textbf{K}=0,$ $\quad\textbf{S}\neq0$;
  \item $A=0=B,$ \hspace{0.5cm} $\textbf{K}\neq0,$ $\quad\textbf{S}=0$,
\end{enumerate}
with classes 1, 2 and 3 satisfying $\textbf{K},\textbf{S} \neq 0$. The spinors belonging to the first three classes are called regular spinors while classes 4, 5 and 6 are labelled as singular spinors  \cite{jcap, Julio, FRJR, Benn, Maj}.Spinors describing fermions in field theory are called Dirac spinors, and they may belong to classes 1, 2 or 3, i.e. all Dirac spinors are necessarily regular ones\footnote{But not all regular spinors are necessarily Dirac spinors, as showed in the reference \cite{CR}.}.

As recently was shown in \cite{bilineares}, due to the adjoint structure of the MDO fermions \cite{1305}, it is extremely necessary to deform the usual Clifford algebra in order to ascertain the right observance of the FPK identities, regarding MDO spinor fields.

\subsection{A short overview on the non-linear Heisenberg theory formalism}\label{subHeisenberg}
The non-linear Heisenberg equation of motion is easily obtained by varying the action with respect to the spinor field, constructed by \cite{heisenbergpaper, novello2}
\begin{eqnarray}
\mathcal{L} = \frac{i}{2}\bar{\psi}^{H}\gamma^{\mu}\partial_{\mu}\psi^{H} - \frac{i}{2}\partial_{\mu}\bar{\psi}^{H}\gamma^{\mu}\psi^{H}-s J_{\mu}J^{\mu},
\end{eqnarray}
thus, non-linear Heisenberg equation reads\footnote{The fundamental field equations must be non-linear in order to represent interaction. The masses of the particles should be a consequence of this interaction \cite{heisenbergpaper1}.} \cite{novello}
\begin{eqnarray}\label{eqheisenberg}
i\gamma^{\mu}\partial_{\mu}\psi^{H} - 2s(A+iB\gamma^5)\psi^H=0,
\end{eqnarray}
where $s$ stands for a constant which has dimension of $(length)^2$ and the physical amounts $A$ and $B$ are given in terms of the usual bilinear covariants associated with Heisenberg spinor, given by \eqref{Azao} and \eqref{Bzao}, respectively.
%\begin{eqnarray}
%A=\bar{\psi}^H\psi^H,
%\end{eqnarray}
%and 
%\begin{eqnarray}
%B=i\bar{\psi}^H\gamma^5\psi^H.
%\end{eqnarray}
The Heisenberg spinor can be represented by a line in a two-dimensional plane ($\pi$), where each axis is represented by the left-hand and right-hand spinors \cite{novello}. In such a way that, the Heisenberg spinor can be portrayed as the following identity
\begin{eqnarray}\label{heisenbergdecomposto}
\psi^H &=& \psi^{H}_{L}+ \psi^{H}_{R},
\end{eqnarray}
in other words,
\begin{eqnarray}
\psi^H=\frac{1}{2}(\mathbbm{1}+\gamma^5)\psi^H+\frac{1}{2}(\mathbbm{1}-\gamma^5)\psi^H.
\end{eqnarray}

A particular class of solutions of the Heisenberg equation \eqref{eqheisenberg} is given by  
\begin{eqnarray}\label{inomatacond}
\partial_{\mu}\psi = (aJ_{\mu}+bK_{\mu}\gamma^5)\psi,
\end{eqnarray}
with $a$, $b$ $\in \mathbb{C}$ of dimensionality $(length)^2$, $J_{\mu}$ and $K_{\mu}$ are covariant and irrotational currents.
A $\psi$ that satisfies the condition \eqref{inomatacond} also satisfies the Heisenberg equation of motion if $a$ and $b$ are such that $2s=i(a-b)$ \cite{novello} and shall be called as RIM (restricted Inomata-McKinley) spinor. As recently was shown in \cite{RIM} every Dirac spinor written in terms of RIM spinors belongs to the class $1$ within Lounesto Classification. In order that \eqref{inomatacond} be integrable, the constants $a$ and $b$ must obey the constraint $Re(a)=Re(b)$.

Hence, we are able to define $J^2=J_{\mu}J^{\mu}$ and consequently 
\begin{eqnarray}
J_{\mu}=\partial_{\mu}S, 
\end{eqnarray}
where 
\begin{eqnarray}\label{eqS}
S=\frac{1}{(a+\bar{a})}\ln \sqrt{J^2}, 
\end{eqnarray}
represents a scalar, and similarly we can write 
\begin{eqnarray}\label{eqK}
K_{\mu}=\partial_{\mu}R,
\end{eqnarray}
with 
\begin{eqnarray}\label{eqR}
R=\frac{1}{(b-\bar{b})}\ln\bigg(\frac{A-iB}{\sqrt{J^2}}\bigg),
\end{eqnarray}
also being a scalar\footnote{In order to make the notation compact, we define $\sqrt{J^2}\equiv J$}. 
From \eqref{inomatacond}, we obtain for the left-hand and right-hand Heisenberg spinors
\begin{eqnarray}
&&\partial_{\mu}\psi^{H}_{L} = (aJ_{\mu}+bK_{\mu})\psi^H_L, \\
&&\partial_{\mu}\psi^{H}_{R} = (aJ_{\mu}-bK_{\mu})\psi^H_R.
\end{eqnarray}
Thus, to complete the program to be accomplished in the scope of this work, one is able to write an arbitrary spinor field, $\psi$, in terms of a $\pi$-plane decomposed Heisenberg spinor 
\begin{eqnarray}\label{psiarbitrario}
\psi=e^{F}\psi^H_L+e^{G}\psi^H_R,
\end{eqnarray}  
and then, looking towards to write a linear theory in terms of a non-linear theory, one analyses the properties encoded on the functions $F$ and $G$ in order to the spinor \eqref{psiarbitrario} satisfy the Dirac equation. This is the prescription used in reference \cite{novello} to write Dirac spinors in terms of RIM-spinors. We will follow this idea in the next Section in order to also write MDO spinors in terms of RIM-spinors.

\section{Mass-dimension-one fermions and RIM-spinors}\label{MDORIM}
Analogously as developed in \cite{novello}, we analyse the possibility to write a MDO fermionic field \cite{1305} in terms of the non-linear Heisenberg spinors. All the discussion is based on two fundamental equations, the non-linear Heisenberg equation and the \emph{Dirac-like} equation for MDO fermions \cite{jcap}, which reads
\begin{eqnarray}\label{diraclike}
(i\gamma^{\mu}\partial_{\mu}\Xi\pm m\mathbbm{1})\lambda_{h}^{S/A}(\boldsymbol{x})=0, 
\end{eqnarray}
where the subscript $h$ stands for the helicity $h=\lbrace\pm,\mp\rbrace$, the upperindex $S/A$ stands for the self-conjugated and anti-self-conjugated spinors, respectively, under action of the charge conjugation operator ($\mathcal{C}\lambda_{h}^{S}= + \lambda_{h}^{S}$ and $\mathcal{C}\lambda_{h}^{A}= - \lambda_{h}^{A}$), and the operator $\Xi$ in its matricial form is given by \cite{bilineares}
\begin{eqnarray}\label{ximatrix}
\Xi= \left(\begin {array}{cccc} {\frac {ip\sin \theta }{m}}&{\frac {-i \left( E+p\cos\theta 
\right) {{\rm e}^{-i\phi}}}{m}}&0&0\\ \noalign{\medskip}{\frac {i
\left( E-p\cos\theta\right) {{\rm e}^{i\phi}}}{m}}&
{\frac {-ip\sin\theta}{m}}&0&0\\ \noalign{\medskip}0&0
&{\frac {-ip\sin\theta}{m}}&{\frac {-i \left( E-p\cos\theta\right) {{\rm e}^{-i\phi}}}{m}}
\\ \noalign{\medskip}0&0&{\frac {i \left( E+p\cos\theta\right) {{\rm e}^{i\phi}}}{m}}&{\frac {ip\sin\theta}{m}}\end {array} \right),
\end{eqnarray}
where $p=|\boldsymbol{p}|$. Then, we obtain the identity
\begin{eqnarray}
\lambda^{S/A}_h = \frac{1}{2}(\mathbbm{1}+\gamma^5)\lambda^{S/A}_{h}+\frac{1}{2}(\mathbbm{1}-\gamma^5)\lambda^{S/A}_{h}.
\end{eqnarray}
 Expliciting the left- and right-handed components,
\begin{eqnarray}
&&\lambda^{S/A}_{R_h}=\frac{1}{2}(\mathbbm{1}-\gamma_5)\lambda^{S/A}_{h}, \\
&&\lambda^{S/A}_{L_h}=\frac{1}{2}(\mathbbm{1}+\gamma_5)\lambda^{S/A}_{h}.
\end{eqnarray} 
We are now able to initiate the process to reach the decomposition (or representation) of the MDO spinors in terms of RIM-spinors, following the ideas of the Subsection \ref{subHeisenberg}. Firstly, one can write
\begin{eqnarray}
\lambda^{S/A}_{h} = e^{\stackrel{\neg}{F}}\psi_{L_{h^{\prime}}}^{H}+ e^{\stackrel{\neg}{G}}\psi_{R_{h^{\prime}}}^{H},
\end{eqnarray}
and, consequently, for the left- and right-handed components, we obtain
\begin{eqnarray}
&&\lambda^{S/A}_{L_h}= e^{\stackrel{\neg}{F}}\psi_{L_{h^{\prime}}}^{H},\\
&&\lambda^{S/A}_{R_h}= e^{\stackrel{\neg}{G}}\psi_{R_{h^{\prime}}}^{H}. 
\end{eqnarray}
The symbol ``\;$\stackrel{\neg}{}$\;" over $F$ and $G$, although commonly used to represent the dual of $\lambda$, is here simply to denote the functions related to $\lambda$ in the attempt to RIM-decompose such a spinor, and do not have any relation to the dual of the field.
%From the algebraic point of view, the sum of the spinors of certain classes did not prescribe the class, that is to say the resulting spinor does not necessarily belongs to the same class of spinors of which it is composed.

Following the program, the next step is to find the explict form of $\stackrel{\neg}{F}$ and $\stackrel{\neg}{G}$ in order that $\lambda^{S/A}_{h}$ satisfies \eqref{diraclike}. By the same akin reasoning presented in \cite{novello} but now for the MDO spinors, we note that 
\begin{eqnarray}\label{chain}
\partial_{\mu} = \partial_{\mu}S\frac{\partial}{\partial S} + \partial_{\mu}R\frac{\partial}{\partial R}.
\end{eqnarray}
Taking into account the relations in equations \eqref{eqS} and \eqref{eqR}, we are able to write \eqref{chain} in this fashion
\begin{eqnarray}
\partial_{\mu}= J_{\mu}\frac{\partial}{\partial S} + K_{\mu}\frac{\partial}{\partial R},
\end{eqnarray}
therefore, one obtains
\begin{eqnarray}\label{elkomaoesq}
\partial_{\mu} \lambda^{S/A}_{L_h} = \bigg(\frac{\partial \stackrel{\neg}{F}}{\partial S}J_{\mu}+ \frac{\partial \stackrel{\neg}{F}}{\partial R}K_{\mu}\bigg)\lambda^{S/A}_{L_h}+(aJ_{\mu}+bK_{\mu})\lambda^{S/A}_{L_h}, 
\end{eqnarray}
and
\begin{eqnarray}\label{elkomaodir}
\partial_{\mu} \lambda^{S/A}_{R_h} = \bigg(\frac{\partial \stackrel{\neg}{G}}{\partial S}J_{\mu}+ \frac{\partial \stackrel{\neg}{G}}{\partial R}K_{\mu}\bigg)\lambda^{S/A}_{R_h}+(aJ_{\mu}-bK_{\mu})\lambda^{S/A}_{R_h}.
\end{eqnarray}
Taking advantage of the \emph{Dirac-like} equation, we multiply the equations \eqref{elkomaoesq} and \eqref{elkomaodir} by $i\gamma^{\mu}$, then, using the fact that $\Xi^2=\mathbbm{1}$ and $[\Xi,\gamma^{\mu}p_{\mu}]=0$, so we have\footnote{The authors choose to work in abstract only with the $\lambda^{S}_{h}$ spinors since the physical content holds the same for all the other MDO spinors, one differing from the other only by a constant phase.} 
\begin{eqnarray}
i\gamma^{\mu}\partial_{\mu}\lambda^{S}_{h} &=& i(A-iB)\bigg(\frac{\partial \stackrel{\neg}{F}}{\partial S}- \frac{\partial \stackrel{\neg}{F}}{\partial R} +(a-b)\bigg)\lambda^{S}_{R_h}+ i(A+iB)\bigg(\frac{\partial \stackrel{\neg}{G}}{\partial S}- \frac{\partial \stackrel{\neg}{G}}{\partial R} +(a-b)\bigg)\lambda^{S}_{L_h}\nonumber\\
&=&m\Xi\lambda^{S}_{h}.
\end{eqnarray}
Using the relations\footnote{For more informations, please, check the appendix.} \eqref{A5}-\eqref{A8}, one obtains the following set of equations:
\begin{eqnarray}
&&\bigg[(A-iB)\bigg(\frac{\partial \stackrel{\neg}{F}}{\partial S}- \frac{\partial \stackrel{\neg}{F}}{\partial R} +(a-b)\bigg)\mathbbm{1}+im\Xi_1\bigg]\lambda^{S}_{R_h}=0, \\
&&\bigg[(A+iB)\bigg(\frac{\partial \stackrel{\neg}{G}}{\partial S}+ \frac{\partial \stackrel{\neg}{G}}{\partial R} +(a-b)\bigg)\mathbbm{1}+im\Xi_2\bigg]\lambda^{S}_{L_h}=0.
\end{eqnarray}
At this stage, we freely summarized the notation and rewrite \eqref{ximatrix} as it follows
\begin{eqnarray}
\Xi= \left( \begin{array}{cc}
\Xi_1  & 0_{2\times 2} \\ 
0_{2\times 2} & \Xi_2 
\end{array} \right).
\end{eqnarray}
After a bit of straightforward calculation, the solutions for $\stackrel{\neg}{F}(S,R)$ and $\stackrel{\neg}{G}(S,R)$ functions are given by 
%\begin{eqnarray}
%\stackrel{\neg}{F}(R,S)= (a-b)R -\frac{p\sin\theta(A+iB)e^{-2(a+\bar{a})S}}{2(a+\bar{a})}, \\
%\stackrel{\neg}{G}(R,S)= -(a-b)R +\frac{p\sin\theta(A+iB)e^{-2(a+\bar{a})S}}{2(a+\bar{a})},
%\end{eqnarray}
%and
%\begin{eqnarray}
%\stackrel{\neg}{F}(R,S)= (a-b)R +\frac{p\sin\theta(A+iB)e^{-2(a+\bar{a})S}}{2(a+\bar{a})}, \\
%\stackrel{\neg}{G}(R,S)= -(a-b)R -\frac{p\sin\theta(A+iB)e^{-2(a+\bar{a})S}}{2(a+\bar{a})}.
%\end{eqnarray}
%Let us define
%\begin{eqnarray}
% \stackrel{\neg}{F}_{\pm}(R,S) & \equiv & (a-b)R \pm \frac{p\sin{\theta} (A+iB) e^{-2(a+\bar{a})S}}{2(a+\bar{a})},\\
% \stackrel{\neg}{G}_{\pm}(R,S) & \equiv & -(a-b)R \pm \frac{p\sin{\theta} (A+iB) e^{-2(a+\bar{a})S}}{2(a+\bar{a})}.
%\end{eqnarray}
\begin{eqnarray}
 \stackrel{\neg}{F}_{\pm}(S,R) & \equiv & -2isR \pm \frac{p\sin{\theta} (A+iB) e^{-2(a+\bar{a})S}}{2(a+\bar{a})},\\
 \stackrel{\neg}{G}_{\pm}(S,R) & \equiv & +2isR \pm \frac{p\sin{\theta} (A-iB) e^{-2(a+\bar{a})S}}{2(a+\bar{a})}.
\end{eqnarray}
Note that 
\begin{equation}
 A+iB = \frac{J^2}{A-iB},
\end{equation}
and from \eqref{eqS}, we have
\begin{equation}
J^{2} = e^{2(a+\bar{a})S}.
\end{equation}
Therefore,
\begin{eqnarray}
 e^{\stackrel{\neg}{F}_{-}} &=& \exp{\left[ -2isR - \frac{1}{2}\frac{p\sin{\theta}}{(a+\bar{a})(A-iB)} \right] },\\
 e^{\stackrel{\neg}{G}_{-}} &=& \exp{\left[ +2isR - \frac{1}{2}\frac{p\sin{\theta}}{(a+\bar{a})(A+iB)} \right] },
\end{eqnarray}
then, with $\vartheta \equiv e^{2isR}$, we have
\begin{eqnarray}
 e^{\stackrel{\neg}{F}_{-}} &=& \frac{1}{\vartheta} \exp{\left[ - \frac{1}{2}\frac{p\sin{\theta}}{(a+\bar{a})(A-iB)} \right]},\\
 e^{\stackrel{\neg}{G}_{-}} &=& \vartheta \exp{\left[ - \frac{1}{2}\frac{p\sin{\theta}}{(a+\bar{a})(A+iB)} \right] }.
\end{eqnarray}
Following an analogue prescription, we can write
\begin{eqnarray}
 e^{\stackrel{\neg}{F}_{\pm}} & = & \frac{1}{\vartheta} \exp{\left[ \pm \frac{1}{2}\frac{p\sin{\theta}}{(a+\bar{a})(A-iB)} \right] },\\
 e^{\stackrel{\neg}{G}_{\pm}} & = & \vartheta \exp{\left[ \pm \frac{1}{2}\frac{p\sin{\theta}}{(a+\bar{a})(A+iB)} \right] }.
\end{eqnarray}
In this manner, we finally write the MDO spinors in terms of RIM-spinors
\begin{equation}
 \lambda_h = \frac{1}{\vartheta} \exp{\left[ \pm \frac{p\sin{\theta}}{2(a+\bar{a})(A-iB)} \right]} \psi^H_{L_h} + \vartheta \exp{\left[ \pm \frac{p\sin{\theta}}{2(a+\bar{a})(A+iB)} \right] }\psi^H_{R_h},
\end{equation}
or, one is able to write the last expression in the fashion (replacing $p$ to $m$)
%\begin{equation}
%\lambda^{S/A}_{h} =  \alpha_h\bigg[ \bigg(\frac{A-iB}{J}\bigg)^{\frac{a-b}{b-\bar{b}}} \exp{\left[ \pm \frac{p\sin{\theta}}{2(a+\bar{a})(A-iB)} \right]} \psi^H_L + \bigg(\frac{J}{A-iB}\bigg)^{\frac{a-b}{b-\bar{b}}} \exp{\left[ \pm \frac{p\sin{\theta}}{2(a+\bar{a})(A+iB)} \right] }\psi^H_R\bigg].
%\end{equation}
\begin{eqnarray}\label{elkorim}
\lambda = \bigg(\sqrt{\frac{J}{A-iB}}\bigg)^{\rho} \exp{\left[ \pm \frac{m\sin{\theta}}{4Re(a)(A-iB)} \right]} \psi^H_L + \bigg(\sqrt{\frac{A-iB}{J}}\bigg)^{\rho} \exp{\left[ \pm \frac{m\sin{\theta}}{4Re(a)(A+iB)} \right] }\psi^H_R,
\end{eqnarray}
where we have defined $\rho \equiv \frac{Im(a)-Im(b)}{Im(b)} = \frac{-2s}{\text{Im}(b)}$. Note that we omitted the upper index $S/A$ due to the fact that such spinors differs from a global phase. As a net result we reach that the MDO fields can be freely represented as a combination of RIM-spinors which satisfy the non-linear Heisenberg equation.

\section{Two-dimensional spinor-spaces: the spinor-plane}\label{spinorplane}

We start this Section giving the definition of the spaces in which we will work on.

\begin{definition}\label{def1}
 We denote by $\Pi^H$ the two-dimensional space whose the set $\mathcal{B} = \{ \Psi^H_L, \Psi^H_R \}$ (namely, the left- and right-handed components of the RIM-spinor $\Psi^H$) forms a basis. Analogously, we denote the spaces $\Pi^D$ (with basis $\mathcal{D} = \{\Psi^D_L,\Psi^D_R\}$ being formed by the components of the Dirac-RIM spinor) and $\Pi^M$ (with basis formed by the MDO-RIM components $\mathcal{M} = \{\lambda_L,\lambda_R\}$). These spaces will be called spinor-planes.
\end{definition}

 In order to achieve better organization, let us record that we can write Dirac spinors \cite{novello} $\Psi^D$ and MDO spinors $\lambda$ in the $\Pi^H$ space, via basis $\mathcal{B}$, as

\begin{eqnarray}
 \Psi^D & = & \exp{\left[\frac{iM}{(a+\bar{a})J}\right]} J^{2\sigma} \left(\sqrt{\frac{J}{A-iB}}\Psi^H_L + \sqrt{\frac{A-iB}{J}}\Psi^H_R \right), \\
 \lambda & = & \exp{\left[\frac{\pm m \sin{\theta}}{4\text{Re}(a)(A-iB)}\right]} \left(\sqrt{\frac{J}{A-iB}}\right)^{\rho} \Psi^H_L + \exp{\left[\frac{\pm m \sin{\theta}}{4\text{Re}(a)(A+iB)}\right]} \left(\sqrt{\frac{A-iB}{J}}\right)^{\rho} \Psi^H_R,
\end{eqnarray}

\noindent with $J^{2\sigma} = \exp{\{ \left[2is - \frac{1}{2}(b-\bar{b})\right] S\}} = \exp{\left[ -i\frac{\text{Im}(a)}{2\text{Re}(a)}\ln{J} \right]}$. Now we will set the following notations for these complex numbers, for the sake of clarity:

% One can easily see, using the Definition (\ref{def1}), that $\Psi^H = (1,1)_\mathcal{B}, \Psi^D = (1,1)_\mathcal{D}, \lambda = (1,1)_\mathcal{M}$.

\begin{eqnarray}
 \alpha & \equiv & \exp{\left[\frac{iM}{(a+\bar{a})J}\right]}, \\
 \beta & \equiv & J^{2\sigma}, \\
 \delta & \equiv & \sqrt{\frac{J}{A-iB}}, \\
 \epsilon & \equiv & \left(\sqrt{\frac{J}{A-iB}}\right)^\rho, \\
 \omega & \equiv & \exp{\left[\frac{\pm m \sin{\theta}}{4\text{Re}(a)(A-iB)}\right]},\\
 \zeta & \equiv & \exp{\left[\frac{\pm m \sin{\theta}}{4\text{Re}(a)(A+iB)}\right]}.
\end{eqnarray}

In this fashion, one can denote the left- and right-handed components of the fields as

\begin{eqnarray}
 \Psi^D_L & = & \alpha \beta \delta \Psi^H_L,\\
 \Psi^D_R & = & \alpha \beta \delta^{-1} \Psi^H_R,\\
 \lambda_L & = & \epsilon \omega \Psi^H_L,\\
 \lambda_R & = & \epsilon^{-1} \zeta \Psi^H_R,
\end{eqnarray}

which leads to

\begin{eqnarray}
 \lambda_L = \chi_1 \Psi^D_L,\\
 \lambda_R = \chi_2 \Psi^D_R,\\
 \Psi^D_L = \chi_1^{-1} \lambda_L,\\
 \Psi^D_R = \chi_2^{-1} \lambda_R,
\end{eqnarray}

\noindent with the coefficients defined as $\chi_1 \equiv \epsilon \omega \delta^{-1} \beta^{-1} \alpha^{-1}$ and $\chi_2 \equiv \epsilon^{-1} \zeta \delta \beta^{-1} \alpha^{-1}$ being obviously invertible. These coefficients and their inverses are the tools that map Dirac-RIM spinors into MDO-RIM spinors and vice-versa. After some straightforward calculations, one achieves an explict form of those complex coefficients as

%OBS.: SE FOR USAR ESSAS EQS DE NOVO TEM QUE ARRUMAR UMA COISA: TÁ FALTANDO UM "S" MULTIPLICANDO OS DOIS PRIMEIROS FATORES DA EXP!
% \begin{eqnarray}
%  \chi_1 & = & \left(\sqrt{\frac{J}{A-iB}}\right)^{\rho-1} \exp{\left[ \frac{1}{2}(b-\bar{b}) -2is -\frac{iM}{(a+\bar{a})J} \pm \frac{m \sin{\theta}}{4\text{Re}(a)(A-iB)} \right]}, \\
%  \chi_1^{-1} & = & \left(\sqrt{\frac{A-iB}{J}}\right)^{\rho-1} \exp{\left[ -\frac{1}{2}(b-\bar{b}) +2is +\frac{iM}{(a+\bar{a})J} \mp \frac{m \sin{\theta}}{4\text{Re}(a)(A-iB)} \right]}, \\
%  \chi_2 & = & \left(\sqrt{\frac{A-iB}{J}}\right)^{\rho-1} \exp{\left[ \frac{1}{2}(b-\bar{b}) -2is -\frac{iM}{(a+\bar{a})J} \pm \frac{m \sin{\theta}}{4\text{Re}(a)(A+iB)} \right]}, \\
%  \chi_2^{-1} & = & \left(\sqrt{\frac{J}{A-iB}}\right)^{\rho-1} \exp{\left[ -\frac{1}{2}(b-\bar{b}) +2is +\frac{iM}{(a+\bar{a})J} \mp \frac{m \sin{\theta}}{4\text{Re}(a)(A+iB)} \right]}.
% \end{eqnarray}

\begin{eqnarray}
 \chi_1 & = & \left(\sqrt{\frac{J}{A-iB}}\right)^{\rho-1} \exp{ \left\{ \frac{1}{2\text{Re}(a)} \left[ \pm \frac{m \sin{\theta}}{2(A-iB)} -i\left( \text{Im}(a)\ln{J} +\frac{M}{J} \right) \right] \right\}}, \\
 \chi_1^{-1} & = & \left(\sqrt{\frac{A-iB}{J}}\right)^{\rho-1} \exp{ \left\{ \frac{1}{2\text{Re}(a)} \left[ \mp \frac{m \sin{\theta}}{2(A-iB)} +i\left( \text{Im}(a)\ln{J} +\frac{M}{J} \right) \right] \right\}}, \\
 \chi_2 & = & \left(\sqrt{\frac{A-iB}{J}}\right)^{\rho-1} \exp{ \left\{ \frac{1}{2\text{Re}(a)} \left[ \pm \frac{m \sin{\theta}}{2(A+iB)} -i\left( \text{Im}(a)\ln{J} +\frac{M}{J} \right) \right] \right\}}, \\
 \chi_2^{-1} & = & \left(\sqrt{\frac{J}{A-iB}}\right)^{\rho-1} \exp{ \left\{ \frac{1}{2\text{Re}(a)} \left[ \mp \frac{m \sin{\theta}}{2(A+iB)} +i\left( \text{Im}(a)\ln{J} +\frac{M}{J} \right) \right] \right\}}.
\end{eqnarray}

This way, one can obtain

\begin{eqnarray}
 \lambda & = & \frac{1}{2}\left[ \chi_1 (\mathbbm{1}+\gamma^5) + \chi_2(\mathbbm{1}-\gamma^5) \right] \Psi^D,\\
 \Psi^D & = & \frac{1}{2}\left[ \chi_1^{-1} (\mathbbm{1}+\gamma^5) + \chi_2^{-1} (\mathbbm{1}-\gamma^5) \right] \lambda.
\end{eqnarray}

If we define the matrices $M \equiv \frac{1}{2}\left[ \chi_1 (\mathbbm{1}+\gamma^5) + \chi_2(\mathbbm{1}-\gamma^5) \right]$ and $N \equiv \frac{1}{2}\left[ \chi_1^{-1} (\mathbbm{1}+\gamma^5) + \chi_2^{-1} (\mathbbm{1}-\gamma^5) \right]$, it easily verifies that $MN = NM = \mathbbm{1}$, i.e., $N = M^{-1}$. Then, we have just proved the following:

\begin{lemma}\label{DiracMDO}
 Let $\varphi_D \in \Pi^D$ and $\varphi_\lambda \in \Pi^M$. There exists a linear isomorphism $M:\Pi^D \rightarrow \Pi^M$, given by means of a matricial operator $M = \frac{1}{2}\left[ \chi_1 (\mathbb{I}+\gamma^5) + \chi_2(\mathbbm{1}-\gamma^5) \right]$, such that
 
\begin{eqnarray}
 \varphi_\lambda & = & M \varphi_D,\\
 \varphi_D & = & M^{-1} \varphi_\lambda.
\end{eqnarray}
 
\end{lemma}

%Given the physical constraint that Dirac spinors has mass dimension $\frac{3}{2}$, it is clear that $\Pi^D \cap \Pi^M = \{ 0 \}$.

\textit{Lemma} \ref{DiracMDO} shows a linear bijective (algebraic) map between special classes of MDO and Dirac fields, when both are decomposable in terms of RIM-spinors.

Note that an analogue procedure can be done between all the other combinations of the spinor-spaces. Thus, using $(v,w)_\mathcal{A}$ as a notation for the coordinates of a given spinor in a basis $\mathcal{A}$ of a spinor-space $\Pi^\mathcal{A}$, for $\mathcal{A} \in \{ \mathcal{B,D,M} \}$, one can represent $\Psi^H$, $\Psi^D$ and $\lambda$ as

\begin{alignat}{3}
 \Psi^H & = (1,1)_\mathcal{B} &{} = {}& (\alpha^{-1} \beta^{-1} \delta^{-1}, \alpha^{-1} \beta^{-1} \delta)_\mathcal{D} &{} = {}& (\epsilon^{-1} \omega^{-1}, \epsilon \zeta^{-1})_\mathcal{M},\\
 \Psi^D & = (\alpha \beta \delta,\alpha \beta \delta^{-1})_\mathcal{B} &{} = {}& (1,1)_\mathcal{D} &{} = {}& (\chi_1^{-1},\chi_2^{-1})_\mathcal{M},\\
 \lambda & = (\epsilon \omega, \epsilon^{-1} \zeta)_\mathcal{B} &{} = {}& (\chi_1, \chi_2)_\mathcal{D} &{} = {}& (1,1)_\mathcal{M}.
\end{alignat}

Precisely, the construction of the (invertible) operators $L:\Pi^H \rightarrow \Pi^D$ and $Q:\Pi^H \rightarrow \Pi^M$ leads to matricial representations given by

\begin{eqnarray}
 \label{matrizL} L & = & \frac{1}{2}\left[ (\alpha \beta \delta) (\mathbbm{1}+\gamma^5) + (\alpha \beta \delta^{-1}) (\mathbbm{1}-\gamma^5) \right],\\
 \label{matrizQ} Q & = & \frac{1}{2}\left[ (\epsilon \omega) (\mathbbm{1}+\gamma^5) + (\omega^{-1} \zeta) (\mathbbm{1}-\gamma^5) \right],
\end{eqnarray}

\noindent such that 

\begin{eqnarray}
 \label{loke1} \Psi^D & = & L \Psi^H,\\
 \label{loke2} \Psi^H & = & L^{-1} \Psi^D,\\
 \label{loke3} \lambda & = & Q \Psi^H,\\
 \label{loke4} \Psi^H & = & Q^{-1} \lambda.
\end{eqnarray}
 
Then, we can state the following:

\begin{lemma}\label{isos}
 Suppose the existence of a spinor-plane $\Pi^S$ with basis formed by left- and right-handed components of a given spinor $\psi = \psi_L + \psi_R$. If $\psi$ can be decomposed in terms of at least one of $\Psi^H$, $\Psi^D$ or $\lambda$ components with both coefficients non vanishing (in other words, the decomposition is invertible), then it can be written in terms of any of those spinors, i.e., $\Pi^S \cong \Pi^H \cong \Pi^D \cong \Pi^M$.
\end{lemma}

\begin{proof}
 It is trivial, using the results of \textit{Lemma} \ref{DiracMDO} and Equations (\ref{loke1} - \ref{loke4}). 
\end{proof}

Note that \textit{Lemma} \ref{DiracMDO} is a corollary of \textit{Lemma} \ref{isos}.

Another fact that is worthwhile to mention is that $M$, $Q$ and $L$ as shown in \textit{Lemma} \ref{DiracMDO} and Equations (\ref{matrizL}) and (\ref{matrizQ}) are all diagonal (as, obviously, their inverses). This is because of the nature of the chirality projector operators, and we can define:

\begin{definition}
 We define $\mathfrak{M}$ as being the space of all matricial operators $R$ such that $\psi = R \varphi$, with $\psi, \varphi$ being spinors that may be decomposed in terms of RIM-spinors. The space $\mathfrak{M}$ has the set of projector operators $\left\{ \frac{1}{2}(\mathbbm{1}+\gamma^5), \frac{1}{2}(\mathbbm{1}-\gamma^5) \right\}$ as basis, working with complex coefficients to form elements of $\mathfrak{M}$, i.e.,
 
\begin{equation}
 \forall R \in \mathfrak{M}, \exists c_1, c_2 \in \mathbb{C} : R = c_1 \frac{1}{2}(\mathbbm{1}+\gamma^5) + c_2 \frac{1}{2}(\mathbbm{1}-\gamma^5).
\end{equation}
 
 Explicitly, $R = \text{diag}(c_2, c_2, c_1, c_1)$.
 
\end{definition}

It should be clear that, when $c_1,c_2 \neq 0$, every $R \in \mathfrak{M}$ is invertible, with $\text{diag}(c_2^{-1}, c_2^{-1}, c_1^{-1}, c_1^{-1}) = R^{-1} \in \mathfrak{M}$.

Finally, given the aspect of all those spinor-planes, we can understand them as being, in fact, exactly the same space, with the matrices $M,L,Q$ and their inverses being change-of-basis matrix operators between the basis $\mathcal{B}$, $\mathcal{D}$ and $\mathcal{M}$, with this being valid for every matrix $R \in \mathfrak{M}$ with other basis of the spinor-plane. This way, we can understand the space $\mathfrak{M}$ as being the space of all change-of-basis matrix operators in the spinor-plane. Then, we have found a two-dimensional space of all spinors that may be decomposed in terms of RIM-spinors (given its left- and right-handed components to form a basis on this space), equipped with a space of change-of-basis matrix operators.

\section{The $s$-space, the spinor-plane and homotopic functions}\label{homotopic}

The reference \cite{novello} analyses carefully the domain of parameters $a$ and $b$, in order to avoid singularities on the potentials $S$ and $R$. Writing the complex numbers $a = a_0 e^{i\phi_1}$ and $b = b_0 e^{i\phi_2}$ in their polar forms, one is able to separate all the possible values for these complex numbers into (only) six disjoint domains:
\begin{eqnarray*}
 \Omega_1 & \equiv & W_1 \otimes Z_1, \\
 \Omega_2 & \equiv & W_4 \otimes Z_1, \\
 \Omega_3 & \equiv & W_4 \otimes Z_4, \\
 \Omega_4 & \equiv & W_2 \otimes Z_2, \\
 \Omega_5 & \equiv & W_3 \otimes Z_2, \\
 \Omega_6 & \equiv & W_3 \otimes Z_3,
\end{eqnarray*}

\noindent in which the intervals are defined as $W_1 = \left(0,\frac{\pi}{2}\right), \; W_2 = \left(\frac{\pi}{2}, \pi\right), \; W_3 = \left(\pi,\frac{3\pi}{2}\right), \; W_4 = \left(\frac{3\pi}{2},2\pi\right)$ for $\phi_1$, with an analogue definition for $Z_1, Z_2, Z_3$ and $Z_4$ as intervals of $\phi_2$.
The point here is that for different choices of $a$ and $b$ in those six domains, one can construct different spinor configurations. Then, in order to make more clear our explanations, we define the s-space:

\begin{definition}
 Let $\Omega \equiv \bigcup_{i=1}^{6} \Omega_i$ be the space of all the feasible choices of parameters $(\phi_1,\phi_2)$ for $a = a(\phi_1)$ and $b = b(\phi_2)$ that define the Heisenberg constant $s = \frac{i(a-b)}{2}$ for the RIM solution (\ref{inomatacond}) of the the Heisenberg equation (\ref{eqheisenberg}). We will call $\Omega$ the s-space.
\end{definition}

Now we are able to introduce another interpretation for the two-dimensional spinor space, as we are dealing in this work: fix a basis on this space, say $\mathcal{B}$, so we are in the ``RIM-copy" of the spinor-plane. In this copy, the spinor $\Psi^H$ is a linear function: accurately, it is the identity function $y_H (x) = x$. Yet in this copy of the plane, we have $\Psi^D$ given by the function $y_D (x) = \left( \frac{A-iB}{J} \right) x$, and $\lambda$ given by $ y_\lambda (x) = \left( \frac{A-iB}{J} \right)^{-2\rho} \omega^{-1} \zeta x$. For both Dirac and MDO cases, we have the variable $x$ being defined via the $s$-space $\Omega$, i.e., $x = x(a(\phi_1),b(\phi_2))$, and also $y = y(a(\phi_1),b(\phi_2))$. But once a pair $(\phi_1, \phi_2) \in \Omega$ is fixed, all coordinates on the spinor-plane for every spinor is a pair $(x,y(x))$ in every basis.

In other words, Dirac, MDO and RIM spinors (depending on which basis we are working on the spinor-plane) are implicit functions of $a$ and $b$ (or, via s-space, of $\phi_1$ and $\phi_2$), i.e., behave like functions of the type

\begin{eqnarray}\label{spinorfunction1}
 \varphi_\mathcal{B}:\Omega & \longrightarrow & \Pi^H \nonumber \\
 (\phi_1,\phi_2) & \mapsto & (f_1,f_2)_\mathcal{B},
\end{eqnarray}

\noindent with $f_1,f_2$ being complex functions of the pair $(\phi_1,\phi_2) \in \Omega$. In a similar way, we can define $\varphi_\mathcal{D}:\Omega \rightarrow \Pi^D$ and $\varphi_\mathcal{M}:\Omega \rightarrow \Pi^M$. Of course, it is also valid for every spinor in the spinor-plane\footnote{This is guaranteed by \textit{Lemma} \ref{isos}.}.

It should be clear that both $\Psi^D$ and $\lambda$ are linear functions (also the identity function) when represented in their ``own copies" of the spinor-plane (i.e., when they are written in terms of the basis $\mathcal{D}$ and $\mathcal{M}$ respectively). In fact, it is true for every possible spinor\footnote{i.e., $\Pi^S \ni \psi = (x,x)_\mathcal{S}$ in the spinor-plane.} $\psi$ as described in \textit{Lemma} \ref{isos}. Following this idea, one can think on the basis change being a deformation of the points (which are functions) on the spinor-plane, leading us to the attempt of construction of a homotopy on this space. Before initiate this, we need first to note that each point on the spinor-plane (in any fixed basis) can be written as $(x,y(x))$, with $y:\mathbb{C} \rightarrow \mathbb{C}$. Notice that every $y = y(x)$ is a function on topological spaces, once $x(\phi_1,\phi_2)$ is set.

For us to begin the construction of the homotopy $H$, let, for instance, the Dirac spinor $\Psi^D$ be represented as $(x,f(x))_\mathcal{D} = (x,g(x))_\mathcal{B}$. Then, we know that $f(x) = x$ and $g(x) = \left( \frac{A-iB}{J} \right) x$. Now we need to find a continous map $H:\mathbb{C} \times [0,1] \rightarrow \mathbb{C}$ such that $H(x,0) = f(x)$ and $H(x,1) = g(x)$ for all $x$. Defining

\begin{equation}
 H(x,t) = (1-t)f(x) + tg(x) = \left[ 1 + t\left (\frac{A-iB}{J} - 1 \right) \right]x,
\end{equation}

\noindent we see that it satisfies the conditions, and a remarkable result comes out: for each fixed value $j \in (0,1)$, the function $H(x,j) \equiv H_j(x)$ induces a new representation for $\Psi^D$ as a pair $(x,H_j(x))_\mathcal{A} \in \Pi^A$, which corresponds to an intermediate copy $\Pi^A$ of the spinor-plane or, equivalently, it gives to the spinor-plane a basis $\mathcal{A} = \{\Psi^A_L,\Psi^A_R\}$ corresponding to the definition of an intermediate spinor $\Psi^A$. Noticing this fact, and remembering the result of \textit{Lemma} \ref{isos}, we can state the following:

\begin{theorem}\label{family}
 Let $f=f(x)$ and $g=g(x)$ be functions such that $(x,f(x))_{\mathcal{A}_0}$ and $(x,g(x))_{\mathcal{A}_1}$ represent the same spinor on the spinor-plane by different basis $\mathcal{A}_0$ and $\mathcal{A}_1$. Then it is possible to construct a homotopy $H(x,t)$ between $f$ and $g$ that defines an infinite family of spinors $\Psi^{A_j}$ (that can be decomposed in terms of RIM-spinors), with each spinor being represented by the identity function if the basis $\mathcal{A}_j = \{ \Psi^{A_j}_L, \Psi^{A_j}_R \}$ is used (or equivalently, each spinor is represented by $(x,x)_{\mathcal{A}_j}$) for each fixed $t = j \in [0,1]$.
\end{theorem}
 
 It is clear that one could construct all the spinors in the spinor-plane by just simply choosing a pair of complex numbers\footnote{With these numbers depending on the values on the s-space.} $(c_1, c_2)$ and writing down, for instance, $\psi = c_1 \psi^H_L + c_2 \Psi^H_R = (c_1, c_2)_\mathcal{B}$, but using the result of the Theorem \ref{family} one can deform continously the functions-coordinates between two specific spinors, instead of just choose, without any criteria, complex numbers as being the coordinates. In other words, this method provides a family of spinors which is related to each other by functions belonging to the same homotopy class.
 
 Moreover, we can state another result, now concerning on homotopy and spinors on a fixed basis:
 
\begin{theorem}\label{homotopyspinors}
 There exists a homotopic equivalence relation between any two spinors $\psi$ and $\varphi$ that can be written in terms of RIM-spinors.
\end{theorem}
 
 \begin{proof}
 Looking at Equation (\ref{spinorfunction1}), what happens is that for every pair $(\phi_1,\phi_2) \in \Omega$ we fix a complex number $x = f_1$, and then another complex number $f_2 = f_2(f_1) = f_2(x)$ is determined to form the pair $(f_1,f_2)$ which represents a spinor in a certain basis of the spinor-plane. Therefore, with the s-space $\Omega$ acting as a support-space to construct the set of all allowed complex numbers $x$, we can understand each spinor $\psi$ itself as a map
 
\begin{eqnarray}\label{spinorfunction2}
 \psi:\mathbb{C} & \longrightarrow & \mathbb{C}^2 \nonumber \\
 x & \mapsto & (x,y(x)).
\end{eqnarray}
 
 Then, we can construct a homotopy $G_\mathcal{A}:\mathbb{C} \times [0,1] \rightarrow \mathbb{C}^2$ between two spinors $\psi = (x,y_\psi)_\mathcal{A}$ and $\varphi = (x,y_\varphi)_\mathcal{A}$ in a fixed basis $\mathcal{A}$, given by
 
 \begin{equation}
  G_\mathcal{A}(x,t) = (x,(1-t)y_\psi + ty_\varphi).
 \end{equation}

 Clearly, in a fixed basis $\mathcal{A}$, we have $G_\mathcal{A}(x,0) = (x,y_\psi) = \psi$ and $G_\mathcal{A}(x,1) = (x,y_\varphi) = \varphi$. Thereupon, $G_\mathcal{A}$ makes explicit an equivalence relation between the spinors $\psi$ and $\varphi$ themselves.
 \end{proof}
 
 Again, we could simply construct ``by hand" spinors in the spinor-plane defining points with the help of s-space $\Omega$, but the comprehensive result of Theorem \ref{homotopyspinors} allows us to obtain representations $(x,y(x))$ of spinors, in a given basis, that are intermediary deformations of two known (homotopic related) spinors, i.e., one can think of equivalence homotopy classes of spinors.
 
 One can remember the well known proposition which states that, if $A$ is a convex subset of $\mathbb{R}^n$ and $X$ is any topological space, then any two continuous maps $f, g: X \rightarrow A$ are homotopic. Thus, the two theorems presented here show a particular (yet remarkable) result, namely, that these $f$ and $g$, with convenient choice of spaces $X$ and $A$, can represent spinor fields.
 
 We will discuss these two Theorems more deeply in the last Section.
 
%  Moreover, looking at Equation (\ref{spinorfunction1}), what happens is that for every pair $(\phi_1,\phi_2) \in \Omega$ we fix a complex number $x = f_1$, and another complex number $f_2 = f_2(f_1) = f_2(x)$ is determined to form the pair $(f_1,f_2)$ which represents a spinor in a certain basis of the spinor-plane. Therefore, with the s-space $\Omega$ acting as a support-space to construct the set of all allowed complex numbers $x$, we can understand each spinor $\psi$ itself as a function
%  
% \begin{eqnarray}\label{spinorfunction2}
%  \psi:\mathbb{C} & \longrightarrow & \mathbb{C}^2 \nonumber \\
%  x & \mapsto & (x,y(x)),
% \end{eqnarray}
%  
%  \noindent then the homotopy $H$ built by the result of the Theorem \ref{family} induces a homotopy $G:\mathbb{C} \times [0,1]$ between two spinors $\psi = (x,y_\psi)_\mathcal{A}$ and $\varphi = (x,y_\varphi)_\mathcal{A}$ in a given basis $\mathcal{A}$, given by $G(x,t) = (x,(1-t)y_\psi + ty_\varphi)$. Thereupon, $G$ makes explicit an equivalence relation between the spinors $\psi$ and $\varphi$ themselves, which allows us to compose spinors and create other equivalence classes of homotopic spinors.

\section{On the Dirac dual, bilinear covariants and Lounesto classification}\label{bilinears}

\subsection{RIM-spinors and bilinear covariants}

 It is well known \cite{novello, RIM} that RIM-spinors $\Psi^H$ are necessarily regular spinors, otherwise it would be possible to have $A = 0 = B$ (with $A = \bar{\Psi^H}\Psi^H$ and $B = i \bar{\Psi^H} \gamma^5 \Psi^H$) and then the Heisenberg non-linear equation would reduce to the ordinary linear Dirac equation. Because of that, it seems that the possibility to have only one of the bilinears $A = 0$ or $B = 0$ (i.e., type-2 and type-3 RIM-spinors) is perfectly feasible, since it remains intact the non-linear aspect of the Heisenberg equation. However,
 
\begin{lemma}\label{RIMtype1}
 The RIM-spinors $\Psi^H$ are necessarily type-1 in Lounesto classification (i.e., $A,B \neq 0$).  
\end{lemma}

\begin{proof}
 Firstly, notice that $R = \left( b - \bar{b} \right)^{-1} \ln{ \left( \frac{A-iB}{J} \right) }$, which means, in particular, that $(b - \bar{b}) \neq 0$. In fact, by \eqref{eqK}, otherwise, we would end up with $K_{\mu}\rightarrow \infty$, leading to an unphysical result. Now, in reference \cite{novello}, it is claimed that $J_\mu$ and $K_\mu$ constitute a basis for vectors constructed by the derivative $\partial_\mu$ operating on functionals of $\Psi^H$, and one has the following equations (which are valid for every $\mu$):
 
 \begin{eqnarray}
  \label{del-A} \partial_\mu A  & = & (a + \bar{a}) A J_\mu + i(b - \bar{b}) B K_\mu, \\
  \label{del-B} \partial_\mu B & = & (a + \bar{a}) B J_\mu + i(b - \bar{b}) A K_\mu.
 \end{eqnarray}

 Suppose that $A = 0$ and $B \neq 0$. Then, Equation (\ref{del-A}) gives $0 = i(b - \bar{b}) B K_\mu$, which is a contradiction, since $\Psi^H$ is a regular spinor and we cannot have $K_\mu = 0, \; \forall \mu$. Thus, $\Psi^H$ cannot be type-3. Moreover, if we suppose that $A \neq 0$ and $B = 0$, then Equation (\ref{del-B}) provides $0 = i(b - \bar{b}) A K_\mu$, an analogue contradiction, and we conclude that $\Psi^H$ cannot be type-2, by the same reason as before. Therefore, we conclude that $A,B \neq 0$ and $\Psi^H$ is a regular type-1 spinor.
\end{proof}
 
 We can extract more informations about $A$ and $B$ from the explicit form of the scalar $R$. In fact, we can note that $(A-iB) \neq 0$ and $J \equiv \sqrt{J^2} \neq 0$. Then, remembering that $J^2 = (A-iB)(A+iB)$, we also conclude that $(A+iB) \neq 0$. Now, let us represent a RIM-spinor as $\Psi^H = (\Psi_{11} \;\; \Psi_{12} \;\; \Psi_{21} \;\; \Psi_{22})^T$. Let us define
 
 \begin{eqnarray}
 \label{A1} A_1 \equiv \Psi_{21}^* \Psi_{11} + \Psi_{22}^* \Psi_{12},\\
 \label{A2} A_2 \equiv \Psi_{11}^* \Psi_{21} + \Psi_{12}^* \Psi_{22},
 \end{eqnarray}
with $r^*$ denoting the complex conjugate of $r \in \mathbb{C}$. Then it is straightforward to see that
 
 \begin{eqnarray}
  A & = & A_1 + A_2,\\
  B & = & i (-A_1 + A_2).
 \end{eqnarray}

 With this in hands, since $A, B \neq 0$, we conclude that $A_1 \neq \pm A_2$. Besides, $A + iB = 2A_1 \neq 0 \Rightarrow A_1 \neq 0$, and $A - iB = 2A_2 \neq 0 \Rightarrow A_2 \neq 0$.
 
 These conditions will turn into strong constraints on RIM-decomposable bilinear covariants.
 
 \subsection{RIM-decomposable spinors and bilinear covariants}

 Let $\psi = \psi_L + \psi_R$ be a spinor that can be decomposed in terms of RIM-spinors, so there exists a matrix $R \in \mathfrak{M}$ such that $\psi = R \Psi^H$. In this case, we can write $R = \text{diag}(r_1,r_1,r_2,r_2)$ with decomposition $\psi = r_1 \Psi^H_L + r_2 \Psi^H_R$. Suppose that its dual is constructed in the Dirac fashion $\bar{\psi} = \psi^\dagger \gamma^0$. Let us represent the bilinear covariants associated to $\psi$ as $A_\psi, B_\psi, \textbf{J}_\psi, \textbf{K}_\psi, \textbf{S}_\psi$.

 We want to categorize all RIM-decomposable spinors in the Lounesto classification. In order to do that, initially we need to know the conditions for $\textbf{J}_\psi \neq 0$, because it is an imposition in the aforementioned classification. Since $\psi = R \Psi^H$, we can write

\begin{equation}
 J_\psi^\mu = (\Psi^H)^\dagger R^\dagger \gamma^0 \gamma^\mu R \Psi^H.
\end{equation}

 Representing $\Psi^H = (\Psi_1 \; \Psi_2)^T$ with $\Psi_j \equiv (\Psi_{j1} \; \Psi_{j2})^T$ for $j \in \{1,2\}$, we obtain

\begin{eqnarray}
 \label{J0} J_\psi^0 & = & |\Psi_1|^2 |r_1|^2 + |\Psi_2|^2 |r_2|^2,\\
 \label{J1} J_\psi^1 & = & -|r_1|^2 (\Psi_{12}^* \Psi_{11} + \Psi_{11}^* \Psi_{12} ) + |r_2|^2 (\Psi_{22}^* \Psi_{21} + \Psi_{21}^* \Psi_{22} ),\\
 \label{J2} J_\psi^2 & = & i \left[ -|r_1|^2 (\Psi_{12}^* \Psi_{11} - \Psi_{11}^* \Psi_{12} ) + |r_2|^2 (\Psi_{22}^* \Psi_{21} - \Psi_{21}^* \Psi_{22} ) \right],\\
 \label{J3} J_\psi^3 & = & -|r_1|^2 (|\Psi_{11}|^2 - |\Psi_{12}|^2) + |r_2|^2 (|\Psi_{21}|^2 - |\Psi_{22}|^2).
\end{eqnarray}
 One has to look for the conditions that lead to $J_\psi^\mu = 0, \; \forall \mu \in \{0,1,2,3\}$, simultaneously. These conditions will form the exactly conditions that we have to avoid. We have to verify the components $J_\psi^\mu$ one by one. Thus, as a start, in order to reach $J_\psi^0 = 0$, one finds three options:

\begin{itemize}
 \item[($i$)] $r_1 \neq 0$, $r_2 = 0$ and $|\Psi_1|^2 = 0$ (which, by symmetry, is equivalent to $r_2 \neq 0$, $r_1 = 0$ and $|\Psi_2|^2 = 0$).
 \item[($ii$)] $|\Psi_1|^2 = 0 = |\Psi_2|^2$.
 \item[$(iii)$] $r_1 = 0 =r_2$.
\end{itemize}

Obviously, $r_1 = 0 =r_2$ is not an allowed option, as it leads to $\psi = 0$ with all bilinear covariants vanishing, which is not interesting. Note that the option $(ii)$ leads to $\Psi^H = 0 = \psi$, then we descart it. Now we have to analyse the case of option $(i)$. In fact, one can easily verifies that condition $(i)$ simultaneously vanishes Equations (\ref{J0}-\ref{J3}), i.e.,

\begin{equation}\label{condJi1}
 \textbf{J}_\psi = 0 \Leftrightarrow (i).
\end{equation}
Then, we conclude that we have to avoid condition $(i)$.

 The scalar $A_\psi = \bar{\psi} \psi$ and the pseudo-scalar $B_\psi = i \bar{\psi} \gamma^5 \psi$ can both be written in terms of the four components of $\Psi^H$, as

\begin{eqnarray}
 \label{Adirac} A_\psi & = & (r_1 r_2^*) (\Psi_{21}^* \Psi_{11} + \Psi_{22}^* \Psi_{12}) + (r_1^* r_2) (\Psi_{11}^* \Psi_{21} + \Psi_{12}^* \Psi_{22}),\\
 \label{Bdirac} B_\psi & = & i \left[ -(r_1 r_2^*) (\Psi_{21}^* \Psi_{11} + \Psi_{22}^* \Psi_{12}) + (r_1^* r_2) (\Psi_{11}^* \Psi_{21} + \Psi_{12}^* \psi_{22}) \right].
\end{eqnarray}
For the particular case of $r_1, r_2 \in \mathbb{R}$ (in other words, if $R$ is real), we have the interesting fact $A_\psi = ( r_1 r_2 ) A$ and $B_\psi = ( r_1 r_2 ) B$, i.e., $A_\psi \propto A$ and $B_\psi \propto B$, and we have that $\psi = r_1 \Psi^H_L + r_2 \Psi^H_R$ is always a type-1 spinor when $r_1, r_2 \in \mathbb{R} - \{0\}$.

 In order to have $\psi$ a RIM-decomposable spinor, we have two options\footnote{Remember that it is equivalent to $r_2 \neq 0$ and $r_1 = 0$.}: $r_1, r_2 \neq 0$ or $r_1 \neq 0$, $r_2 = 0$ ($r_2 \neq 0$, $r_1 = 0$). Now, note that the second option cannot happen with $|\Psi_1|^2 = 0$ ($|\Psi_2|^2 = 0$) occuring, since it would lead to condition $(i)$. Then, for the sake of clarity, we will separate our study in two cases. On the case $r_1 \neq 0$, $r_2 = 0$, we will show that
 
\begin{lemma}\label{JKS}
 For a RIM-decomposable spinor $\psi$ such that $\bar{\psi} = \psi^\dagger \gamma^0$, we have

\begin{equation}
 \textbf{J}_\psi \neq 0 \Rightarrow \left( \textbf{K}_\psi \neq 0 \;\text{and}\; \textbf{S}_\psi = 0 \right),
\end{equation}
everytime the conditions $r_1 \neq 0$ and $r_2 = 0$ (or, equivalently, $r_2 \neq 0$ and $r_1 = 0$) are satisfied.
\end{lemma}

\begin{proof}
 We will look for the conditions to make $\textbf{K}_\psi = 0$ and $\textbf{S}_\psi = 0$ in this case. Firstly, let us analyse $\textbf{K}_\psi$. Analogously to what was made to reach Equations (\ref{J0}-\ref{J3}), we obtain

\begin{equation}\label{K0}
 K_\psi^0 = |\Psi_1|^2 |r_1|^2 - |\Psi_2|^2 |r_2|^2.
\end{equation}

 Now, suppose $r_2 = 0$ (then, $r_1 \neq 0$). Thus, in order to have $K_\psi^0 = |\Psi_1|^2 |r_1|^2 = 0$, one must have $|\Psi_1|^2 = 0$. But it would lead to the condition $(i)$. Therefore, by relation (\ref{condJi1}), in this case we cannot have $\textbf{K}_\psi = 0$.

% Without loss of generalization, suppose $r_2 = 0$ (then, $r_1 \neq 0$). Then, in order to have $K_\psi^0 = |\Psi_1|^2 |r_1|^2 = 0$, one must have $|\Psi_1|^2 = 0$. But it would lead to $J_\psi^0 = 0$, a contradiction. So in this case we have $\textbf{K}_\psi \neq 0$.

% But there is another way to have $K_\psi^0 = 0$, which is $|\Psi_1|^2 |r_1|^2 = |\Psi_2|^2 |r_2|^2$. In this case, we have necessarily $r_1, r_2 \neq 0$, in order to keep $J_\psi^0 \neq 0$. But this implies that $A_\psi, B_\psi \neq 0$ by Equations (\ref{Adirac}) and (\ref{Bdirac}), then $\psi$ is a regular spinor, and $\textbf{K}_\psi, \textbf{S}_\psi \neq 0$.

% But there is another way to have $K_\psi^0 = 0$, which is $|\Psi_1|^2 |r_1|^2 = |\Psi_2|^2 |r_2|^2$, with necessarily $r_1, r_2 \neq 0$. In order to keep $J_\psi^0 \neq 0$, note that this implies $|\Psi_1|^2, |\Psi_2|^2 \neq 0$, which is equivalent to $|\Psi_{j1}|^2, |\Psi_{j2}|^2 \neq 0$, for $j \in \{1,2\}$. It also implies that $|\Psi_1|^2 = \frac{|r_1|^2}{|r_2|^2} |\Psi_2|^2$. If we use this last equality in Equation (\ref{K1}), we have $K_\psi^1 = 2 (|\Psi_{12}|^2 |r_1|^2 - |\Psi_{21}|^2 |r_2|^2 )$, which never vanishes. Then, again, $\textbf{K}_\psi \neq 0$.

Now, let us analyse $\textbf{S}_\psi$. In the same fashion, we can write

\begin{eqnarray}
 \label{S01} S_\psi^{01} & = & -i \left[ (r_2^* r_1) (\Psi_{22}^* \Psi_{11} + \Psi_{21}^* \Psi_{12}) - (r_1^* r_2) (\Psi_{12}^* \Psi_{21} + \Psi_{11}^* \Psi_{22}) \right],\\
 \label{S02} S_\psi^{02} & = & (r_2^* r_1) (\Psi_{22}^* \Psi_{11} - \Psi_{21}^* \Psi_{12}) - (r_1^* r_2) (\Psi_{12}^* \Psi_{21} - \Psi_{11}^* \Psi_{22}),\\
 \label{S03} S_\psi^{03} & = & -i \left[ (r_2^* r_1) (\Psi_{21}^* \Psi_{11} - \Psi_{22}^* \Psi_{12}) - (r_1^* r_2) (-\Psi_{11}^* \Psi_{21} + \Psi_{12}^* \Psi_{22}) \right],\\
 \label{S12} S_\psi^{12} & = & (r_2^* r_1) (\Psi_{21}^* \Psi_{11} - \Psi_{22}^* \Psi_{12}) + (r_1^* r_2) (\Psi_{11}^* \Psi_{21} - \Psi_{12}^* \Psi_{22}),\\
 \label{S13} S_\psi^{13} & = & -i \left[ (r_2^* r_1) (\Psi_{22}^* \Psi_{11} - \Psi_{21}^* \Psi_{12}) + (r_1^* r_2) (\Psi_{12}^* \Psi_{21} - \Psi_{11}^* \Psi_{22}) \right],\\
 \label{S23} S_\psi^{23} & = & (r_2^* r_1) (\Psi_{22}^* \Psi_{11} + \Psi_{21}^* \Psi_{12}) + (r_1^* r_2) (\Psi_{12}^* \Psi_{21} + \Psi_{11}^* \Psi_{22}).
\end{eqnarray}

Again, without loss of generalization, suppose $r_2 = 0$ (and, so, $r_1 \neq 0$). In this case, it is obvious that we always have $S_\psi^{\mu \nu} = 0$, i.e., $\textbf{S}_\psi = 0$, and it does not depend on any condition for the components $\Psi_{ij}$ whatsoever.

 Summarizing, what we have found is that $\textbf{J}_\psi \neq 0 \Rightarrow \textbf{K}_\psi \neq 0$ and $\textbf{J}_\psi \neq 0 \Rightarrow \textbf{S}_\psi = 0$, when $r_1 \neq 0$ and $r_2 = 0$, which can be written as $\textbf{J}_\psi \neq 0 \Rightarrow \left( \textbf{K}_\psi \neq 0 \;\text{and}\; \textbf{S}_\psi = 0 \right)$. This ends the proof.
\end{proof}

 Now, note that, if $r_1 \neq 0$ and $r_2 = 0$, then $A_\psi = 0 = B_\psi$: in other words, in this case we are dealing necessarily with singular spinors. But, using \textit{Lemma} \ref{JKS}, we have that in order to classify these spinors on the spinor-plane by the Lounesto classification, if $r_1 \neq 0$ and $r_2 = 0$, then we are dealing with type-6 singular spinors (see Subsection \ref{subLounesto}).

 What about other situations that could lead to singular RIM-decomposable spinors? Let us see. In fact, the other option left is to have $r_1, r_2 \neq 0$. Note that, looking at the definitions (\ref{A1}) and (\ref{A2}) and Equations (\ref{Adirac}) and (\ref{Bdirac}), it is straightforward to see that we can write
 
\begin{eqnarray}
 A_\psi & = & (r_1 r_2^*) A_1 + (r_1^* r_2) A_2,\\
 B_\psi & = & i \left[ -(r_1 r_2^*) A_1 + (r_1^* r_2) A_2 \right].
\end{eqnarray}
Then, as a last try, if we choose $r_1, r_2 \neq 0$, we see that we cannot have simultaneously $A_\psi = 0 = B_\psi$, thus $\psi$ is a regular spinor, and $\textbf{K}_\psi, \textbf{S}_\psi \neq 0$. Indeed, if we impose $A_\psi = 0 = B_\psi$ with $r_1, r_2 \neq 0$, then $(r_1 r_2^*) A_1 = -(r_1^* r_2) A_2$ and $(r_1 r_2^*) A_1 = +(r_1^* r_2) A_2$, which has no solution once we know that $A_1, A_2 \neq 0$. Note that this implies that we cannot have type-4 and type-5 singular RIM-decomposable spinors at all, since this exhausts all possibilities for $r_1$ and $r_2$ in the construction of a non-null $\psi$. Moreover, we have shown here that a decomposition leading to a singular spinor needs to satisfy $r_1 \neq 0$ and $r_2 = 0$, and having $r_1 \neq 0$ and $r_2 = 0$ is sufficient in order to have a decomposition leading to a singular spinor. Therefore, we can state the following:

\begin{proposition}\label{bilinearDirac1}
 Suppose $\psi = R \Psi^H$, with $R = \text{diag}(r_1,r_1,r_2,r_2) \in \mathfrak{M}$ such that $\psi = r_1 \Psi^H_L + r_2 \Psi^H_R$, satisfying $\textbf{J}_\psi \neq 0$ and $\bar{\psi} = \psi^\dagger \gamma^0$. Then, the statements below are equivalent:
 \begin{itemize}
  \item[$(i)$] $\psi$ is a singular spinor.
  \item[$(ii)$] $\psi$ is a type-6 spinor.
  \item[$(iii)$] $r_1 = 0$ or $r_2 = 0$ (but not both).
  \item[$(iv)$] $\psi$ is projected only in $\left(0,\frac{1}{2}\right)$ or $\left(\frac{1}{2},0\right)$ representation, i.e., $\psi \propto \Psi^H_L$ or $\psi \propto \Psi^H_R$.
 \end{itemize}
\end{proposition}

 We have realized, thus, that once we are setting $r_1, r_2 \neq 0$ we are dealing with regular spinors. In this case, we have other three options to verify: type-1 ($A_\psi, B_\psi \neq 0$,) type-2 ($A_\psi \neq 0$ and $B_\psi = 0$) and type-3 ($A_\psi = 0$ and $B_\psi \neq 0$). Since we know that $\Psi^H$ itself and the Dirac RIM-decomposable field $\Psi^D$ are both type-1 \cite{RIM}, we only have to check the other two possibilities. First, if we set $A_\psi = 0$, then $(r_1 r_2^*) A_1 = -(r_1^* r_2) A_2 \neq 0$, and we can write $B_\psi = 2i(r_1^* r_2) A_2 = -2i(r_1 r_2^*) A_1 \neq 0$, then it is possible to have type-2 RIM-decomposable spinors. Analogously, if we set $B_\psi = 0$, then $(r_1 r_2^*) A_1 = (r_1^* r_2) A_2 \neq 0$, and it leads to $A_\psi = 2i(r_1^* r_2) A_2 = 2i(r_1 r_2^*) A_1 \neq 0$, which means that it is also possible to have type-3 RIM-decomposable spinors. With this in hands, we state that

\begin{proposition}\label{bilinearDirac2}
 Suppose $\psi = R \Psi^H = r_1 \Psi^H_L + r_2 \Psi^H_R$, with $R = \text{diag}(r_1,r_1,r_2,r_2) \in \mathfrak{M}$, satisfying $\textbf{J}_\psi \neq 0$ and $\bar{\psi} = \psi^\dagger \gamma^0$. Then, $\psi$ is a regular spinor if, and only if, $r_1,r_2 \neq 0$. Yet, in this case, we have that:
 
\begin{itemize}
 \item[$(i)$] $\psi$ is a type-1 spinor if, and only if, $\displaystyle A \neq -iB \left( \frac{r_1 r_2^* \pm r_1^* r_2}{r_1 r_2^* \mp r_1^* r_2} \right)$.
 \item[$(ii)$] $\psi$ is a type-2 spinor if, and only if, $(r_1^* r_2)^2 \neq (r_1 r_2^*)^2$ and $\displaystyle A = -iB \left( \frac{r_1 r_2^* + r_1^* r_2}{r_1 r_2^* - r_1^* r_2} \right)$.
 \item[$(iii)$] $\psi$ is a type-3 spinor if, and only if, $(r_1^* r_2)^2 \neq (r_1 r_2^*)^2$ and $\displaystyle A = -iB \left( \frac{r_1 r_2^* - r_1^* r_2}{r_1 r_2^* + r_1^* r_2} \right)$.
\end{itemize}
\end{proposition}

\begin{proof}
  Firstly, we have already seen that, in the hypothesis of this Proposition, the condition $r_1, r_2 \neq 0$ is necessary and sufficient for $\psi$ to be a regular spinor: in fact, this particular result can be understood as a corollary of Proposition \ref{bilinearDirac1}. Now, noticing that $A_1 = \frac{A+iB}{2}$ and $A_2 = \frac{A-iB}{2}$, one is able to write
 
\begin{eqnarray}
 \label{Apsi} A_\psi = \frac{1}{2} \left[ A(r_1 r_2^* + r_1^* r_2) + iB (r_1 r_2^* - r_1^* r_2) \right],\\
 \label{Bpsi} B_\psi = -\frac{i}{2} \left[ A(r_1 r_2^* - r_1^* r_2) + iB (r_1 r_2^* + r_1^* r_2) \right].
\end{eqnarray}
We know that $A_1, A_2 \neq 0$. Yet, one cannot reach $(r_1 r_2^* - r_1^* r_2) = 0 = (r_1 r_2^* + r_1^* r_2)$ with $r_1, r_2 \neq 0$: in fact, it would lead to $A_\psi = 0 = B_\psi$, an unattainable case here, as we have seen.

 Then, in order to reach $B_\psi = 0$, we need to have $A(r_1 r_2^* - r_1^* r_2) + iB (r_1 r_2^* + r_1^* r_2) = 0$. Moreover, one cannot have $(r_1 r_2^* + r_1^* r_2) = 0$ or $(r_1 r_2^* - r_1^* r_2) = 0$ isolated, because it will never make $B_\psi = 0$; in other words, $(r_1^* r_2)^2 \neq (r_1 r_2^*)^2$. Thus, the only option left is $A(r_1 r_2^* - r_1^* r_2) = -iB (r_1 r_2^* + r_1^* r_2)$, which implies that $\displaystyle A = -iB \left( \frac{r_1 r_2^* + r_1^* r_2}{r_1 r_2^* - r_1^* r_2} \right)$. In this case, it is garanteed that $A_\psi \neq 0$. This proves item $(ii)$.
 
 Now, if one wants to have $A_\psi = 0$, for analogue reasons as the case above, we have $(r_1^* r_2)^2 \neq (r_1 r_2^*)^2$ and $A(r_1 r_2^* - r_1^* r_2) = -iB (r_1 r_2^* + r_1^* r_2)$. It leads to $\displaystyle A = -iB \left( \frac{r_1 r_2^* + r_1^* r_2}{r_1 r_2^* - r_1^* r_2} \right)$ (with $B_\psi \neq 0$ garanteed), which proves item $(iii)$.
 
 So far, we have seen that the only way to vanish $A_\psi$ or $B_\psi$ without vanish both at the same time is to have $A(r_1 r_2^* \mp r_1^* r_2) = -iB (r_1 r_2^* \pm r_1^* r_2)$. Then, we conclude that we cannot have these conditions valid in order to keep both $A_\psi, B_\psi \neq 0$, i.e., having $\displaystyle A \neq -iB \frac{(r_1 r_2^* \pm r_1^* r_2)}{(r_1 r_2^* \mp r_1^* r_2)}$ is equivalent to say that $\psi$ can only be type-1, proving item $(i)$.
 
%  Finally, we conclude that $A_\psi, B_\psi \neq 0$ is equivalent to $(r_1^* r_2)^2 = (r_1 r_2^*)^2$. Indeed, in this case we have $(r_1^* r_2) = +(r_1 r_2^*)$ (which implies $A_\psi = A(r_1 r_2^*) \neq 0$ and $B_\psi = B (r_1 r_2^*) \neq 0$) or $(r_1^* r_2) = -(r_1 r_2^*)$ (which implies $A_\psi = iB(r_1 r_2^*) \neq 0$ and $B_\psi = -iA (r_1 r_2^*) \neq 0$), proving item $(i)$.
\end{proof}

 Indeed, $\Psi^H$ has $r_1 = r_2 = 1$, and Proposition \ref{bilinearDirac2} trivially confirms that $\Psi^H$ is type-1, with $A_{\Psi^H} = A$ and $B_{\Psi^H} = B$ being easily obtained by Equations (\ref{Apsi}) and (\ref{Bpsi}), as expected. As another example, for the Dirac spinor $\Psi^D$, we have $r_1 = \alpha \beta \delta$ and $r_2 = \alpha \beta \delta^{-1}$ as defined in Section \ref{spinorplane}, and one can verify that $\frac{r_1+r_2}{r_1-r_2} = -i\left(\frac{A}{B}\right) \Rightarrow -iB\left( \frac{r_1+r_2}{r_1-r_2} \right) = -A \neq A$, and $\frac{r_1-r_2}{r_1+r_2} = i\left(\frac{B}{A}\right) \Rightarrow -iB\left( \frac{r_1-r_2}{r_1+r_2} \right) = \frac{B^2}{A} \neq A$, which confirms that $\Psi^D$ is also type-1.

Summarizing the results of this Section, Propositions \ref{bilinearDirac1} and \ref{bilinearDirac2} provide an easy method to separate all spinors $\psi$ allowed in the spinor-plane\footnote{With dual defined as $\bar{\psi} = \psi^\dagger \gamma^0$.} in the Lounesto classification by just looking at their coefficients $(r_1, r_2)_\mathcal{B}$ in the ``RIM-copy" $\Pi^H$: if both coefficients are non vanishing, then the spinor is regular (with an easy way to verify if it is type-1, type-2 or type-3: simply divide the sum of the coefficients by the difference - and the difference by the sum - and multiply by $-iB$ in order to verify if it is equal to $A$), while if one (and only one) of the coefficients is zero then it is a singular type-6 spinor, with no need to the often hard work of construction of all the bilinear covariants. As we cannot have $r_1 = 0 = r_2$, all feasible cases are contemplated.

\section{Final Remarks}

 The second main result of the reference \cite{RIM}, concerning on exotic spinor fields, allows us to state that all spinors belonging to the spinor-plane that has a dynamic equation are not exotic spinors, i.e., the underlying topology of the space-time $M$ of which these spinors may emerge is trivial, in the sense that it has a trivial fundamental group $\pi_1(M) = 0$. In particular, this spinor-plane accomodates a bijective linear map between special classes (i.e., both being RIM-decomposable and, therefore, non exotic) of MDO and Dirac spinors. This mapping is quite natural, as it uses RIM-spinors as a fundamental element making the mediation between Dirac and MDO fields. Although this mapping has some constraints imposed in the fields themselves (they had to be RIM-decomposable), one does not have to work with the bilinear covariants, which is often a hard situation to deal with when we study MDO spinors, since they do not necessarily fit in the usual Lounesto classification. Therefore, the mapping developed here transcends the problem of Lounesto classification of MDO spinors.

 Among the outcomes of this work, we emphasize that Theorem \ref{family} is an exhaustive result: it gives not only the possibility to write down explicitly all possible spinors that can be decomposed by RIM-spinors, by giving the left- and right-handed components of each of them, but it also makes explicit an equivalence relation (via the homotopy $H$) between all the functions that represent spinor-plane coordinates. On the other hand, the Theorem \ref{homotopyspinors} is another robust result, providing a way to deform spinors in the spinor-plane, enabling the composition and the eventual classification of equivalence classes of homotopic spinors via the homotopies $G_\mathcal{A}$.

 The two Theorems are related, in the sense that both treat the subject of writing down explicit forms of representing the spinors that can be written in terms of RIM-spinors, by showing a homotopic equivalence relation. The main difference between them, which indeed complements each other, lies on the fact that Theorem \ref{family} provides a method to obtain the left- and right- handed components of the spinors by constructing basis for the spinor-plane, while Theorem \ref{homotopyspinors} supplies a way to obtain points (the spinors themselves) in a given fixed basis. In other words, while in Theorem \ref{family} we are continously deforming the spinor-plane itself (obtaining new basis for the space), in Theorem \ref{homotopyspinors} we are continously deforming the points of the plane in a fixed basis (obtaining intermediate spinors between two fixed ones).

 The understanding of the very nature of spinors is a field of study under development, which is as significant in Physics as in Mathematics. Theorems \ref{family} and \ref{homotopyspinors} may be the beginning of a new way to look at the construction of spinors, opening the possibility to the discovery of interesting relations via homotopy theory, which is perhaps one of the most important ideas behind algebraic topology.

 Propositions \ref{bilinearDirac1} and \ref{bilinearDirac2} facilitate the categorization of RIM-decomposable spinors $\psi$, that has $\bar{\psi} = \psi^\dagger \gamma^0$, in the Lounesto classification: they provide a complete and easy way to determine how these spinors are classified in the Lounesto classification when their dual is defined in the Dirac fashion. In fact, they connect the coefficients of their decomposition (or, in other words, their coordinates on the spinor-plane given in the basis $\mathcal{B} = \{ \Psi^H_L, \Psi^H_R \}$) directly with the Lounesto classification, avoiding the construction of all bilinear covariants and the often laborious process of check which of them are null and which ones are not. In particular, Proposition \ref{bilinearDirac2} is a generalization of the \textit{Lemma} $1$ in reference \cite{RIM}, which states that every Dirac spinor decomposable in terms of RIM-spinors is type-1 in the Lounesto classification. %We have shown here that every spinor $\psi$ that satisfies $\bar{\psi} = \psi^\dagger \gamma^0$ and $\psi = r_1 \Psi^H_L + r_2 \Psi^H_R$ for $r_1,r_2 \neq 0$ is type-1, with the Dirac spinor $\Psi^D$ being a particular case with $r_1 = \alpha \beta \delta$ and $r_2 = \alpha \beta \delta^{-1}$ as defined in Section \ref{spinorplane}: really, by the very nature of $r_1$ and $r_2$ for $\Psi^D$, one cannot have $r_1 = 0$ without also $r_2 = 0$ and vice-versa, and then, as long as we cannot have $r_1 = 0 = r_2$, the only feasible option is $r_1, r_2 \neq 0$ which, by Corollary \ref{bilinearDirac2}, means that $\Psi^D$ is indeed type-1.

 It is worthwhile to make clear that the core of all results of this work is in the RIM-decomposition itself, in the sense that the major element that links all \textit{Lemmas}\footnote{The only one which is not related directly to the RIM-decomposition is \textit{Lemma} \ref{RIMtype1}, but it is about the RIM-spinor itself.}, Propositions and Theorems presented here is the pair of coefficients of the decomposition (or, in other words, the coordinates in the spinor-plane) of a given spinor in terms of RIM-spinors. Following this idea, one can study spinor properties in a very similar way if a given spinor is decomposable in terms of another. Thus, this work provides a working protocol that can be useful in other cases of the theory of spinors field of study.
 
 With regard to direct physical applications, this work provides a homotopical method of construction of any possible spinor field allowed in the Spinor Theory of Gravity (STG), which is a theory of gravitation built via a class of solutions of the linearized Einstein equations of General Relativity constructed from RIM-spinors \cite{novello2, novellojcap}, i.e., a gravitation theory with RIM-spinors playing a fundamental role. Moreover, since bilinear covariants are associated with physical observables, we developed a way to easy verify the possible couplings of a particle associated to a given spinor in this theory, by means of their coefficients in the RIM-decomposition.
 
 On what concerns the bijective linear map between Dirac and MDO spinors, one can think of its usefulness directly related to the task of understanding dark matter, which can be described by MDO fields \cite{dark}. Once dark matter interacts very weakly with Standard Model (SM) particles, and aspects of Dirac fields are known in the SM context (in particular the subset of RIM-decomposable Dirac spinors treated here), one can work with this Dirac-MDO mapping in further investigations on extending SM incorporating MDO spinors.

 Further results concerning questions about more properties related to the homotopies in the spinor-plane are under investigation. Moreover, the behaviour of MDO spinors and their bilinear covariants in this space is also a topic under study.

% Further results concerning on questions like if this homotopy class has some dynamical implications depending on the possible dynamic equations of the spinors that ``initiate" and ``finish" the homotopy (or between the spinors that are composed with each other) are under current investigation. In other words, perhaps the homotopy between spinors that have dynamic equations could influentiate the dynamic equations of all other spinors that belongs to the same homotopy class. Future works could contemplate, yet, the case in which the dual of $\psi$ is not defined in the Dirac fashion, which happens to be the case of the MDO spinor $\lambda$. Such topic is also under study.

\section{Acknowledgement}
The authors are grateful to Professor Julio Marny Hoff da Silva for useful conversation. DB thank for CAPES for the financial support, RJBR thank to CAPES and CNPq (Grant Number 155675/2018-4) for the financial support. CHCV thanks to CNPq (Grant No 300381/2018-2) for the financial support.

\appendix
\section{Mass-dimension-one fields and the Fierz-Pauli-Kofink Identities}
As it can be seen in \cite{bilineares}, it is possible to build the basis vectors for the mass-dimension-one spinor's case using the usual Clifford algebra. For any element $\Gamma$ belonging to such algebra, the FPK relation reads
\begin{eqnarray}
(\stackrel{\neg}{\lambda_{h}} \Gamma \gamma_{\mu}\lambda_{h})=(\stackrel{\neg}{\lambda_{h}} \Gamma \lambda_{h})\lambda_{h}-(\stackrel{\neg}{\lambda_{h}} \Gamma\gamma_{5}\lambda_{h})\gamma_{5}\lambda_{h},
\end{eqnarray}
where $\Gamma \in \{ \mathbbm{1}, \gamma_{5}, \gamma_{\mu}, \Xi\gamma_{5}\gamma_{\mu}\Xi \}$. From the above relation we obtain the following:
\begin{eqnarray}
J^2=A^{2}+B^{2},
\end{eqnarray} 
and we also have
\begin{eqnarray}
(\stackrel{\neg}{\lambda_{h}}\Xi\gamma_{5}\gamma_{\mu}\Xi\lambda_{h})\gamma^{\mu}\lambda_{h}&=&(\stackrel{\neg}{\lambda_{h}}\Xi\gamma_{5}\Xi\lambda_{h})\lambda_{h}-(\stackrel{\neg}{\lambda_{h}}\Xi^2\lambda_{h})\gamma_{5}\lambda_{h},\nonumber\\
&=&(\stackrel{\neg}{\lambda_{h}}\gamma_{5}\lambda_{h})\lambda_{h}-(\stackrel{\neg}{\lambda_{h}}\lambda_{h})\gamma_{5}\lambda_{h}.\label{L}
\end{eqnarray}
Note that
\begin{eqnarray*}
[\Xi,\gamma_{5}]=0, \quad \{\gamma_{\mu},\gamma_{5}\}=0 \;\;\;\mbox{and} \quad \Xi^2=\mathbbm{1},
\end{eqnarray*}
with such relations at hands, one is able to write
\begin{eqnarray}
(\stackrel{\neg}{\lambda_{h}}\Xi\gamma_{5}\gamma_{\mu}\Xi\lambda_{h})\gamma^{\mu}\gamma^{5}\lambda_{h}&=&-(\stackrel{\neg}{\lambda_{h}}\Xi\gamma_{5}\gamma_{\mu}\Xi\gamma_{5}\lambda_{h})\gamma^{\mu}\lambda_{h},\nonumber\\
&=&(\stackrel{\neg}{\lambda_{h}}\Xi\gamma_{\mu}\Xi\lambda_{h})\gamma^{\mu}\lambda_{h}.\label{LJ}
\end{eqnarray}

Finally using relations \eqref{L} and \eqref{LJ}, we obtain
\begin{eqnarray}
\label{A5}J_{\mu}\gamma^{\mu}\lambda_{L}&=&(A-iB)\lambda_{R},\\
\label{A6}J_{\mu}\gamma^{\mu}\lambda_{R}&=&(A+iB)\lambda_{L}, \\
\label{A7}K_{\mu}\gamma^{\mu}\lambda_{L}&=&-(A-iB)\lambda_{R},\\
\label{A8}K_{\mu}\gamma^{\mu}\lambda_{R}&=&(A+iB)\lambda_{L}.
\end{eqnarray}

\section{Some comments on the linear and non-linear aspects of the spinor fields}
	We know that the Dirac equation is linear with respect to the spinor fields:
	\begin{equation}\label{equacaodedirac}
	 \left( i \gamma^\mu \partial_\mu - M \right) \Psi^D = 0,
	\end{equation}
	with $M$ being a mass parameter. The solutions $\Psi^D$ are called Dirac spinors, which, because of the linear aspect of the Dirac equation \eqref{equacaodedirac}, are said to be linear spinor fields.
	
	The non-linear counterpart, which we call the Heisenberg equation, is given by
	\begin{eqnarray}\label{equacaodeheisenberg}
	  [i\gamma^{\mu}\partial_{\mu}-2s(A+iB\gamma^{5})]\Psi^{H} = 0,
	\end{eqnarray}
	with $A := \bar{\Psi}^{H}{\Psi}^{H}$ and $B := i\bar{\Psi}^{H}\gamma^{5}{\Psi}^{H}$ being the bilinear covariants associated to the so-called Heisenberg spinor $\Psi^H$, and $s$ is a constant with dimension $(\text{length})^2$. The Heisenberg equation is non-linear with respect to the fields, as one can easily verify by rewriting it as
	\begin{equation}
	 i\gamma^{\mu}\partial_{\mu}\Psi^{H} - 2s \bar{\Psi}^{H}{\Psi}^{H} \Psi^{H} + 2s \bar{\Psi}^{H}\gamma^{5}{\Psi}^{H} \gamma^{5} \Psi^{H} = 0.
	\end{equation}
	Thus, Heisenberg spinors are general solutions of Equation (\ref{equacaodeheisenberg}), and in this sense we say that they are non-linear spinor fields.
	
	Let us compare some aspects of the linear Dirac dynamics and the non-linear Heisenberg dynamics. For the linear case, a particular solution of the Dirac equation (plane waves) can be written as
	\begin{equation}
	 \partial_\mu \Psi^D = i k_\mu \Psi^D.
	\end{equation}
	On the other hand, in the non-linear case of Heisenberg spinors $\Psi^H$, it is possible to find solutions \cite{novello} defined by the property
	\begin{eqnarray}\label{condicaodeinomata}
	  \partial_{\mu} \Psi^H = (aJ_{\mu}+bK_{\mu}\gamma^{5})\Psi^H,
	\end{eqnarray}
	with $a,b \in \mathbb{C}$ such that $2s = i(a-b)$. Yet, the integrability condition forces us to have $\textrm{Re}(a) = \textrm{Re}(b)$. Now, a spinor $\Psi^H$ that satisfies the condition (\ref{condicaodeinomata}) is called a RIM spinor, with RIM standing for Restricted Inomata-McKinley.
	
	The term ``restricted'' was created by reference \cite{RIM}, and the reason is the following: in the original decomposition \cite{akira} the first term of the rhs of Equation (\ref{condicaodeinomata}) is given by $K^\lambda\gamma_\lambda \gamma_\mu \gamma^5$. The mapping between this term and $J_\mu$ is not something straightforward. In fact, it is given by means of a matrix operator that we will call $G$. For the construction of this matrix, we will begin introducing the so-called Inomata-McKinley spinors, that are special Heisenberg spinors which satisfy
	\begin{equation}\label{InomataOriginal}
	 \partial_\mu \Psi = \underbrace{\frac{1}{2} \epsilon}_{\tilde{a}} \left( \bar{\Psi} \gamma^\lambda \gamma_5 \Psi \right) \gamma_\lambda \gamma_\mu \gamma_5 \Psi \underbrace{-2\epsilon}_{\tilde{b}} \left( \bar{\Psi} \gamma_\mu \gamma_5 \Psi \right) \gamma_5 \Psi.
	\end{equation}
	Then,
	\begin{eqnarray}
	 \partial_\mu \Psi & = & \underbrace{\tilde{a} (-i)}_{a} K^\lambda \gamma_\lambda \gamma_\mu \gamma_5 \Psi + \underbrace{\tilde{b}(-i)}_{b} K_\mu \gamma_5 \Psi \Rightarrow \nonumber \\
	 \partial_\mu \Psi & = & a K^\lambda \gamma_\lambda \gamma_\mu \gamma_5 \Psi + b K_\mu \gamma_5 \Psi.
	\end{eqnarray}
	Comparing with Equation \eqref{condicaodeinomata}, one wants that $K^\lambda \gamma_\lambda \gamma_\mu \gamma_5 \Psi := \varepsilon J_\mu \Psi$, implying that $K^\lambda \gamma_\lambda \gamma_\mu \gamma_5 = \varepsilon J_\mu \mathbbm{1} G$, for some operator $G$. Then, making $\varepsilon G \rightarrow G$, one has
	\begin{equation}
	 K^\alpha \gamma_\alpha \gamma_\mu \gamma_5 = J_\mu G.
	\end{equation}
% 	\begin{eqnarray}
% 	 K^\lambda \gamma_\lambda \gamma_\mu \gamma_5 \Psi & := \varepsilon J_\mu \Psi & \Rightarrow \\
% 	 K^\lambda \gamma_\lambda \gamma_\mu \gamma_5 & = \varepsilon J_\mu \mathbb{I} G & \xRightarrow{\varepsilon G \rightarrow G} \nonumber \\
% 	 K^\alpha \gamma_\alpha \gamma_\mu \gamma_5 & = J_\mu G.&
% 	\end{eqnarray}
	Multiplying by $J^\mu$ from the left, we obtain $J^\mu K^\alpha \gamma_\alpha \gamma_\mu \gamma_5 = J^2 G$, which leads us to
	\begin{equation}
	 G =\frac{1}{J^2} \left( J^\mu K^\alpha \gamma_\alpha \gamma_\mu \gamma_5 \right).
	\end{equation}
	But $\{ \gamma_\nu, \gamma_\lambda \} = \eta_{\nu \lambda}$, so we can write $\gamma_\alpha \gamma_\mu = \frac{1}{2} \left[ \gamma_\alpha, \gamma_\mu \right] + \eta_{\mu \alpha}$. Thus,
	\begin{eqnarray}
	 G = \frac{1}{J^2} \left( J^\mu K^\alpha \frac{1}{2}\left[ \gamma_\alpha, \gamma_\mu \right] + \cancelto{_0}{J \bullet K}\;\;\; \right) \gamma_5,
	\end{eqnarray}
	and finally we have the explicit form for the operator that restricts the Inomata-McKinley spinors to the specific case of RIM spinors:
	\begin{eqnarray}
	 G = \frac{1}{2J^2} J^\mu K^\alpha \left[ \gamma_\alpha, \gamma_\mu \right] \gamma_5.
	\end{eqnarray}
	So, RIM spinors are special cases of Inomata-McKinley spinors, and both are particular cases of Heisenberg spinors. All of them are non-linear fields, since their dynamics are conducted by the Heisenberg equation. In a pictorial representation way, we can summarize this as $\Psi^H_{\text{Heisenberg}} \supset \Psi^H_{\text{Inomata-McKinley}} \supset \Psi^H_{\text{RIM}}$.
	
\section{Some comments on the Elko mass-dimension-one spinor fields}	
	
In the early days of mass dimension one spinors, the theory was presented in such a way that a breaking Lorentz term took part in the spin sums. As a net result, the associated quantum field was non-local and there was a preferred axis of symmetry. After all, the theory was shown to be invariant under $SIM(2)$ and $HOM(2)$ transformations \cite{horv}, being then a typical theory carrying the Very Special Relativity symmetries \cite{cohen}. Quite recently, important advances on the spinor dual theory has opened the possibility of circumvent the Weinberg no-go theorem, proposing a spinor field of spin $1/2$ endowed with mass dimension one, local, neutral with respect to gauge interactions, and whose theory respects Lorentz symmetries \cite{1305,Ahluwa2,1602}. We should bring to the scene the canonical Wigner work on the irreducible representations of the Poincar\'e group \cite{Wigner1}. By Poincar\'e group, as usual, it is understood as the semi-simple extension of the orthochronous proper Lorentz group encompassing translations. By investigating the irreducible representations for this case, no particle as a fermion with canonical mass dimension one was found \cite{elko666}.  

Here the situation is different, however, when discrete symmetries are taken into account, i.e., when not only the orthochronous proper group is considered. This point was also analyzed by Wigner, in a less known paper \cite{Wigner2}. Interestingly enough, Wigner found a fermionic irreducible representations whose behaviour under $C, P$ and $T$\footnote{Where $C$, $P$ and $T$ stand for the charge conjugation, parity and time reversal operators, respectively.} are exactly what was expected for the bosonic fields, in other words, Wigner found (theoretically) a class of particles, to be more specifically  fermions, which are endowed with ``bosonical'' character. For concreteness, while conventional wisdom states that fermions belonging to the standard model (quarks and leptons) obey $T^2=-1$ ($(CPT)^2=-1$) and bosons $T^2=+1$ ($(CPT)^2=+1$), Wigner also has shown that, in the very realm of full Poincar\'e symmetries, it is also possible to have $T^2=+1$ for fermions (leading to $(CPT)^2=+1$). It turns out that the MDO field taken into account in this work performs a realization of the (indeed odd) aforementioned fermionic representation, from where we can adduce its ``bosonical'' character \cite{elko666}. 	

In this way, what we want to emphasize is that through the analysis of Wigner's works, all arguments corroborate with what we already know about the characteristics of the MDO fermion in question. Although we are dealing with a fermion, it does not have the same mass dimensionality as the Dirac fermions do, in addition to the fact that it has the quantum field propagator similar to that of the scalar field and respect only the Klein-Gordon equation. That is, the MDO are fermions carrying bosonical aspects.

Regarding its formal structure, the Elko spinors are defined as
\begin{equation}\label{elkospinor}
\lambda^{S/A}_{h}(\boldsymbol{p}) = \left(\begin{array}{c}
\pm i\Theta[\phi_L(\boldsymbol{p})]^* \\ 
\phi_L(\boldsymbol{p})
\end{array}\right),
\end{equation}
with $\Theta$ being the Wigner time-reversal operator, given by
\begin{equation}
 \Theta = \left(\begin{array}{cc}
                 0 & -1\\
                 1 & 0
                \end{array}
\right).
\end{equation}
Note that $\Theta[\phi_L(\boldsymbol{p})]^*$ and $\phi_L(\boldsymbol{p})$ are defined as right-hand and left-hand components (under Lorentz transformations), with the upper index $S/A$ standing for self-conjugated and anti-self-conjugated via charge conjugation operation 
\begin{equation}
\mathcal{C}\lambda^{S/A}_{h}=\pm\lambda^{S/A}_{h},
\end{equation}
while the lower index $h$ represents the helicity of each component.
The dual helicity feature is encoded on the relations 
\begin{eqnarray}\label{helicidadeleft}
\vec{\sigma}.\hat{p}\;\phi_L^{\pm}(\boldsymbol{p})=\pm\phi_L^{\pm}(\boldsymbol{p}),
\end{eqnarray}
and the other component has opposite helicity, i.e.,
\begin{equation}\label{helicidaderight}
\vec{\sigma}.\hat{p}\;\Theta[\phi_L^{\pm}(\boldsymbol{p})]^* = \mp \Theta[\phi_L^{\pm}(\boldsymbol{p})]^*.
\end{equation}
The helicity is simply flipped by the action of the Wigner time-reversal operator.

Referring to dynamics, such spinors do not fulfil Dirac dynamical equation \eqref{equacaodedirac}, due to the fact that their representation spaces are not linked by the parity symmetry, they are related by the Wigner time-reversal operator. This way, MDO spinors dynamics are governed only by the Klein-Gordon equation. The last statement is translated into the mass dimensionality of the refered spinors.

\end{document}